\documentclass[11pt]{article}
\pdfoutput=1
\usepackage[whole]{bxcjkjatype}
\usepackage[utf8]{inputenc}
\usepackage[
 a4paper,
 total={160mm,257mm},
 left=25mm,
 top=20mm]{geometry}
\usepackage{hyperref}
\usepackage{amsfonts}
\usepackage{amsmath}
\usepackage{amssymb}
\usepackage{mathrsfs}
\usepackage{color}
\usepackage{authblk}
\usepackage{physics}
\usepackage{multicol}
\usepackage{tikz}
\usepackage{setspace}
\usepackage{amsthm}

\theoremstyle{definition}
\newtheorem{theorem}{Theorem}[section]
\newtheorem{corollary}[theorem]{Corollary}
\newtheorem{definition}[theorem]{Definition}
\newtheorem{lemma}[theorem]{Lemma}
\newtheorem{proposition}[theorem]{Proposition}
\theoremstyle{remark}

\newtheorem{remark}[theorem]{Remark}

\newcommand{\rem}[1]{\textbf{\textcolor{blue}{(#1)}}}

\newcommand{\rvev}[1]{\langle \, #1 \, \rangle}

\DeclareMathOperator*{\Pf}{Pf}

\definecolor{c1}{RGB}{0,90,187}
\definecolor{c2}{RGB}{255,213,0}

\numberwithin{equation}{section}

\usepackage[root radius=.8mm,edge length=.5cm,mark=o]{dynkin-diagrams}

\title{\textbf{Pfaffian Interaction and $BCD$-quiver Matrix Models}} 
\author{\textsc{Nicolas Babinet} and \textsc{Taro Kimura}}
\affil{Institut de Math\'ematiques de Bourgogne, Universit\'e Bourgogne Franche-Comt\'e}
\date{}

\begin{document}

\maketitle

\begin{abstract}
We study matrix models involving Pfaffian interactions as generalizations of the standard $\beta = 1$ and $\beta = 4$ matrix models.
We present the Pfaffian formulas for the partition function and the characteristic polynomial averages.
We also explore the matrix chain with the Pfaffian interaction, which realizes the $BCD$-type quiver matrix models.
\end{abstract}

\tableofcontents

\section{Introduction}

Pfaffian is an analog of the determinant, which is one of the central concepts in the linear algebra.
In the context of statistical physics and quantum field theory, Pfaffian is frequently utilized to describe real (Majorana) fermions, while the determinant is used for complex (Dirac) fermions.
It has been also known that the Pfaffian structure is often discussed for certain classes of random matrices, called the orthogonal and symplectic ensembles, which are classified into the $\beta = 1$ and $\beta = 4$ models due to the Wigner--Dyson classification.
Similarly to the $\beta = 2$ model (unitary ensemble), which exhibits the determinantal structure, the correlation functions (including the partition function) are in general described using the Pfaffian.%
\footnote{One may also use the quaternion determinant, which is essentially equivalent to the Pfaffian.}
See \cite{Mehta:2004RMT,Forrester:2010} for details.

In this paper, we explore formal eigenvalue integrals, involving various Pfaffian interactions, that we call the generalized $\beta = 1$ and $\beta = 4$ models.
We obtain the Pfaffian formulas for the corresponding partition functions and also the average of characteristic polynomials.
As an application of the Pfaffian interaction, we then construct the matrix chain model, involving several sets of eigenvalues.
In order to classify such a matrix chain, we use the graphical description based on the Dynkin-quiver diagram.
Such a multi-matrix model is often called the quiver matrix model, which has been originally introduced to study its connection with conformal field theory~\cite{Marshakov:1991gc,Kharchev:1992iv}.
For example, the standard linear matrix chain is classified into the type $A$ quiver model,
\begin{equation}
    A_n \ : \ \dynkin{A}{}
    \label{eq:Dynkin_A}
\end{equation}
The remaining classical classes are known as the $BCD$-type,
\begin{equation}
    B_n \ : \ \dynkin{B}{} \hspace{4em}
    C_n \ : \ \dynkin{C}{} \hspace{4em}
    D_n \ : \ \dynkin{D}{}
    \label{eq:Dynkin_BCD}
\end{equation}
In particular, the $B$ and $C$ types involve the oriented arrows (non-simply-laced quivers), so that its realization is not straightforward as for the other simply-laced classes.
We show that such a non-simply-laced quiver matrix model can be constructed using the Pfaffian interaction, and discuss its relation to the simply-laced cases through the folding process based on the outer automorphism of the Dynkin-quiver diagram.

The remaining part of this paper is organizes as follows.
In Section~\ref{sec:preliminary}, we collect the notations and the fundamental properties and formulas regarding the determinant and the Pfaffian, that we use in this paper.
In Section~\ref{sec:Pf}, we introduce the generalized $\beta = 1$ and $\beta = 4$ matrix models involving the Pfaffian interactions, and show the Pfaffian formulas for their partition functions.
In Section~\ref{sec:ch_poly}, we consider the correlation functions of the characteristic polynomials with respect to the generalized $\beta = 1$ and $\beta = 4$ models, and again show the Pfaffian formulas for them.
In Section~\ref{sec:chain}, we explore the matrix chain model, and show that the quiver matrix models classified into the type $BCD$ could be constructed using the Pfaffian interaction.

\subsubsection*{Acknowledgments}

We would like to thank Soichi Okada for a useful communication.
This work was supported by ``Investissements d'Avenir'' program, Project ISITE-BFC (No.~ANR-15-IDEX-0003), EIPHI Graduate School (No.~ANR-17-EURE-0002), and Bourgogne-Franche-Comté region.

\section{Preliminary}\label{sec:preliminary}

In this Section, we present the notations and the preliminary results that we will use in the following parts.

\subsection{Notations}

For $n \in \mathbb{N}$, we define a set of integers,
\begin{align}
    [n] := ( 1, \ldots, n )
    \, .
\end{align}
For a subset $I \subset [n]$, we denote
\begin{align}
    \Sigma(I) = \sum_{i \in I} i
    \, .
\end{align}
Let $X$ be an $n \times m$ matrix. 
For $I \subset [n]$, $J \subset [m]$, we denote by $X(I;J)$ the submatrix of $X$ obtained by picking up rows indexed by $I$ and columns indexed by $J$. 
If $X$ is a skew-symmetric matrix, we write $X(I) = X(I;I)$ for simplicity.

\subsubsection*{Operator formalism}

Let $\dd{\mu(x)}$ be an integral measure for the variable $x$.
We denote a function (wave function) associated with a formal state $\ket{f}$ in the Hilbert space by
\begin{align}
    f(x) = \rvev{ x \mid f }
    \, .
\end{align}
For the integral operator $K$ associated with the corresponding kernel $K(x,y)$, we write
\begin{align}
   \rvev{x \mid K \mid f} = \int \dd{\mu(y)} K(x,y) f(y)
    \, .
\end{align}
In this convention, the kernel of the integral operator $K$ is given by
\begin{align}
    K(x,y) = \rvev{ x \mid K \mid y }
    \, .
\end{align}
Hence, the identity operator is given by the integral over the complete set,
\begin{align}
    \int \dd{\mu(x)} \ket{x} \bra{x} = 1
    \, .
\end{align}
We denote the transposition of the kernel,
\begin{align}
    K^\text{T}(y,x) = \rvev{ y \mid K^\text{T} \mid x }
    \, .
\end{align}
We also write the inner product as follows,
\begin{subequations}
\begin{align}
    \rvev{f \mid g} & = \int \dd{\mu(x)} f(x) g(x)
    \, , \\
    \rvev{ f \mid K \mid g} & = \int \dd{\mu(x)} \dd{\mu(y)} f(x) K(x,y) g(y)
    \, .
\end{align}
\end{subequations}

\subsection{Determinantal formulas}

\subsubsection{Vandermonde determinant}

We first show several determinantal formulas regarding the Vandermonde determinant, which is a main building block of the matrix model.
See also the monographs on this subject \cite{Mehta:2004RMT,Forrester:2010,Eynard:2015aea} for more details.

\begin{proposition}[Vandermonde determinant]
Let $N \in \mathbb{N}$.
For a set of variables $X = (x_i)_{i \in [N]}$, the $N$-variable Vandermonde determinant is given by a rank $N$ determinant,
\begin{align}
    \Delta_N(X) = \prod_{1 \le i < j \le N} (x_i - x_j) = \det_{i,j \in [N]} p_{N-j}(x_i)
    \, ,
    \label{eq:VDMdet}
\end{align}
where $p_{k}(x) = x^k + \cdots$ is a degree-$k$ monic polynomial.
\end{proposition}
Using this expression, we have the following Lemma, which is a crucial identity to discuss the $\beta = 4$ model.
\begin{lemma}\label{lemma:beta=4Vandermonde}
For the $N$-variable Vandermonde determinant, the following determinantal formula holds,
\begin{align}
    \Delta_N(X)^4 = (-1)^{N(N+1)/2} \det_{\substack{i \in [N] \\ j \in [2N]}} \qty[ p_{2N-j}(x_i) \ p_{2N-j}'(x_i) ]
    \, .
\end{align}
\end{lemma}
\begin{proof}
The $2N$-variable Vandermonde determinant of $X = (x_i)_{i \in [N]}$ and $Y = (y_i)_{i \in [N]}$ may be written as follows,
\begin{align}
    \Delta_{2N}(X;Y) = \prod_{\substack{i,j\in[N] \\ i<j}} (x_i - x_j) (y_i - y_j) (x_i - y_j) (x_j - y_i) \prod_{i \in [N]} (x_i - y_i)
    \, .
\end{align}
On the other hand, the $2N$-variable Vandermonde determinant is given by
\begin{align}
    \Delta_{2N}(X;Y) 
    = \det_{\substack{i \in [N] \\ j \in [2N]}} \qty[ p_{2N-j}(x_i) \ p_{2N-j}(y_i) ] 
    = \det_{\substack{i \in [N] \\ j \in [2N]}} \qty[ p_{2N-j}(x_i) \ p_{2N-j}(y_i) - p_{2N-j}(x_i)] 
    \, .
\end{align}
Hence, we obtain
\begin{align}
    \Delta_N(X)^4 & = (-1)^{N(N+1)/2} \lim_{y_i \to x_i} \Delta_{2N}(X;Y) \prod_{i \in [N]} (y_i - x_i)^{-1}
    \nonumber \\
    & = (-1)^{N(N+1)/2} \lim_{y_i \to x_i} \det_{\substack{i \in [N] \\ j \in [2N]}} \qty[ p_{2N-j}(x_i) \ \frac{p_{2N-j}(y_i) - p_{2N-j}(x_i)}{y_i - x_i} ]
    \nonumber \\
    & = (-1)^{N(N+1)/2} \det_{\substack{i \in [N] \\ j \in [2N]}} \qty[ p_{2N-j}(x_i) \ p_{2N-j}'(x_i) ]
    \, .
\end{align}
This completes the proof.
\end{proof}

\subsubsection{Schur polynomials}

In this paper, we shall use the Schur polynomial as the basis of the symmetric polynomial of the formal eigenvalues.
See, e.g.~\cite{Macdonald:2015} for the details on this subject.

\begin{definition}[Schur polynomial]
Let $\lambda$ be a partition, a non-increasing sequence of non-negative integers,
\begin{align}
    \lambda = (\lambda_1 \ge \lambda_2 \ge \cdots \ge \lambda_\ell > \lambda_{\ell+1} = \cdots = 0 )
    \, ,
\end{align}
where $\ell = \ell(\lambda)$ is called the length of the partition.
Denoting the transposed partition by $\lambda^\text{T}$, we have $\ell(\lambda) = \lambda_1^\text{T}$.
We define the size of the partition by
\begin{align}
    |\lambda| = \sum_{i = 1}^\infty \lambda_i = \sum_{i = 1}^\infty \lambda_i^\text{T} = |\lambda^\text{T}|
    \, .
\end{align}
Then, the Schur polynomial of $N$ variables, $X = (x_i)_{i\in[N]}$, is defined as follows,
\begin{align}
    s_\lambda(X) = \frac{1}{\Delta_N(X)} \det_{i,j \in [N]} x_i^{\lambda_j + N - j}
    \, .
    \label{eq:Schur_def}
\end{align}
If $\ell(\lambda) > N$, it always vanishes $s_\lambda(X) = 0$.
We also remark $s_\emptyset(X) = 1$.
\end{definition}

For the latter convenience, we also define a modified version of the Schur polynomial as follows.
\begin{definition}
Let $N \in \mathbb{N}$, and we denote $X = (x_i)_{i \in [N]}$, $Y = (y_i)_{i \in [N]}$.
For a partition $\lambda$, we define 
\begin{align}
    a_\lambda(X;Y) = \det_{i \in [N], j \in [2N]} 
    \begin{bmatrix}
    x_i^{\lambda_j + 2N - j} & (y_i^{\lambda_j + 2N - j})'
    \end{bmatrix}
    \, .
\end{align}
Then, we define a modified Schur polynomial as follows,
\begin{align}
    \mathsf{s}_\lambda(X;Y) = \frac{a_\lambda(X;Y)}{a_\emptyset(X;Y)}
    \, .
    \label{eq:Schur_def_mod}
\end{align}
\end{definition}

Regarding the Schur polynomials, we will use the following decomposition in the study of the characteristic polynomials.
\begin{lemma}[Schur polynomial expansion]\label{lemma:Schur_expansion}
Let $X = (x_i)_{i \in [N]}$ and $Y = (y_i)_{i \in [M]}$.
We have the following expansions with the Schur polynomial,
\begin{subequations}
\begin{align}
    \prod_{\substack{i \in [N], j \in [M]}} (x_i - y_j)
    & = \sum_{\lambda \subseteq (M^N)} (-1)^{|\lambda|} s_{\lambda^\vee} (X) s_\lambda(Y)
    \, , \\
    \prod_{\substack{i \in [N], j \in [M]}} (x_i - y_j)^{-1}
    & = \det_N X^{-M} \sum_{\lambda | \ell(\lambda) \le \operatorname{min}(M,N)} s_\lambda(X^{-1}) s_\lambda(Y)
    \, , \nonumber\\
    & = \det_M Y^{-N} \sum_{\lambda | \ell(\lambda) \le \operatorname{min}(M,N)} s_\lambda(X) s_\lambda(Y^{-1})
    \, ,
\end{align}
\end{subequations}
where we define the dual partition
\begin{align}
    \lambda^\vee = (\lambda_1^\vee,\ldots,\lambda_M^\vee) = (N - \lambda_M^\text{T}, \ldots, N - \lambda_1^\text{T})
    \, ,
\end{align}
and the length of the partition denoted by $\ell(\lambda) = \lambda_1$.
\end{lemma}
\begin{proof}
This follows from the Cauchy sum formula.
See, e.g.,~\cite{Macdonald:2015} for details.
\end{proof}

\begin{lemma}[Cauchy determinant]
We define the Cauchy determinant for $X = (x_i)_{i \in [N]}$ and $Y = (y_i)_{i \in [M]}$ as follows,
\begin{align}
    \Delta_{N|M}(X|Y) := \frac{\Delta_N(X) \Delta_M(Y)}{\prod^{i \in [N]}_{j \in [M]} (x_i - y_j)}
    \, .
    \label{eq:Cauchy_det1}
\end{align}
Then, the following formulas hold.
\begin{subequations}\label{eq:Cauchy_det2}
\begin{align}
    & N \ge M : &
    \Delta_{N|M}(X|Y) & = \sum_{\lambda | \ell(\lambda) \le M} 
    \det_{\substack{i \in [N], j \in [M] \\ k \in [N-M]}} \begin{bmatrix}
    x_i^{- \lambda_j - M + j - 1} \\ x_i^{k-1}
    \end{bmatrix}
    \det_{i,j \in [M]} y_i^{\lambda_j + M - j}
    = 
    \det_{\substack{i \in [N], j \in [M] \\ k \in [N-M]}}
    \begin{bmatrix}
    \displaystyle \frac{1}{x_i - y_j} \\[1em] x_i^{k-1}
    \end{bmatrix}
    \, \\
    & N \le M : &
    \Delta_{N|M}(X|Y) & = \sum_{\lambda | \ell(\lambda) \le N} 
    \det_{i,j \in [N]} x_i^{\lambda_j + N - j}    
    \det_{\substack{i \in [M], j \in [N] \\ k \in [M-N]}} \begin{bmatrix}
    y_i^{- \lambda_j - N + j - 1} \\ y_i^{k-1}
    \end{bmatrix}
    = 
    \det_{\substack{i \in [N], j \in [M] \\ k \in [M-N]}}
    \begin{bmatrix}
    \displaystyle \frac{1}{x_i - y_j} \\[1em] y_j^{k-1}
    \end{bmatrix}
\end{align}
\end{subequations}
\end{lemma}
\begin{proof}
The expressions as the summation over the partitions are obtained from the Schur polynomial expansion (Lemma~\ref{lemma:Schur_expansion}).
In addition, for the case $N \ge M$, the summation over the partition is given by
\begin{align}
    \sum_{\ell(\lambda) \le M} = \sum_{ 0 \le \lambda_M \le \cdots \le \lambda_1 \le \infty} = \sum_{0 \le r_M < \cdots < r_1 \le \infty} = \frac{1}{M!} \sum_{\substack{r_\alpha = 0, \ldots,\infty \\ r_\alpha \neq r_\beta}}
\end{align}
where we write $r_\alpha = \lambda_\alpha + M - \alpha$.
Then, we have
\begin{align}
    \Delta_{N|M}(X|Y) & = \det_{\substack{ i \in [N], j \in [M] \\ k \in [N-M] }} 
    \begin{bmatrix}
    \displaystyle
    \sum_{r=0}^\infty x_i^{-r-1} y_j^r \\
    x_i^{k-1}
    \end{bmatrix}
    = \det_{\substack{ i \in [N], j \in [M] \\ k \in [N-M] }} 
    \begin{bmatrix}
    \displaystyle
    \frac{1}{x_i - y_j} \\
    x_i^{k-1}
    \end{bmatrix}
    \, .
\end{align}
The other case $N \le M$ is similarly analyzed by exchanging the $x$ and $y$ variables.
\end{proof}

\subsection{Pfaffian formulas}

We present several formulas involving Pfaffians used in this paper.
\begin{proposition}
For $n,m \in 2 \mathbb{N}$, let $A$ and $\widetilde{A}$ be invertible skew-symmetric matrices of size $n$ and $m$, respectively.
Let $B$ be an $n \times m$ generic matrix.
Then, the following factorization formula holds for the Pfaffian of the block matrix,
\begin{align}
    \Pf \begin{bmatrix}
    A & B \\ - B^\text{T} & \widetilde{A}
    \end{bmatrix}
    & = 
    \Pf A \Pf \qty[\widetilde{A} + B^\text{T} A^{-1} B]
    = \Pf \widetilde{A} \Pf \qty[A + B \widetilde{A}^{-1} B^\text{T}]
    \, .
    \label{eq:Pf_block}
\end{align}
For $A = \widetilde{A} = 0$, we obtain
\begin{align}
    \Pf \begin{bmatrix}
    0 & B \\ - B^\text{T} & 0
    \end{bmatrix}
    = \det B
    \, .
    \label{eq:Pf2det}
\end{align}
\end{proposition}
\begin{remark}
This formula is a Pfaffian analog of the block matrix determinant formula,
\begin{align}
    \det 
    \begin{bmatrix}
    A & B \\ C & D
    \end{bmatrix}
    = \det A \det (D - C A^{-1} B)
    = \det D \det (A - B D^{-1} C)
    \, .
\end{align}
\end{remark}

\begin{lemma}[Laplace-type expansion~\cite{Okada:2019AM}]
\label{lemma:Laplace_exp}
For an $m \times m$ skew-symmetric matrix $Z$, and an $m \times n$ matrix $W$ with $m > n$, we have the following expansion of the Pfaffian,
\begin{align}
        \Pf
        \begin{bmatrix}
        Z & W \\ - W^\text{T} & 0
        \end{bmatrix}
        = \sum_{\substack{I \subset [n] \\ |I| = m-n}} (-1)^{\Sigma(I) + {m \choose 2}} \Pf Z(I) \det W([m] \backslash I; [n])
\end{align}
where $I$ runs over all $(m-n)$-element subsets of $[n]$.
\end{lemma}

\begin{lemma}[Cauchy--Binet-type formula~\cite{Ishikawa:1995LMA,Okada:2019AM}]
\label{lemma:CB_exp}
Let $l \in \mathbb{Z}_{\ge 0}$, $m,n \in \mathbb{N}$, such that $m-l \in 2 \mathbb{N}$, and $m-l \le n$.
For an $n \times n$ skew-symmetric matrix $Z$, an $m \times n$ matrix $W$, and an $m \times l$ matrix $X$, we have the following expansion of the Pfaffian,
\begin{align}
    \Pf 
    \begin{bmatrix}
    W Z W^\text{T} & X \\ - X^\text{T} & 0
    \end{bmatrix}
    = \sum_{\substack{I \subset [n] \\ |I| = m-l}} \Pf Z(I) \det [ W([m];I) \ X([m];[l]) ]
    \label{eq:CB0}
\end{align}
where $I$ runs over all $(m-l)$-element subsets of $[n]$.
If $l = 0$, there is no contribution of the matrix $X$,
\begin{align}
    \Pf \qty[WZW^\text{T}] = \sum_{\substack{I \subset [n] \\ |I| = m}} \Pf Z(I) \det W([m];I)
    \, .
    \label{eq:CB1}
\end{align}
If $m = n$ (hence $l = 0$), the only possibility is $I = \emptyset$, so that we have
\begin{align}
    \Pf \qty[WZW^\text{T}] = \Pf Z \det W
    \, .
    \label{eq:CB2}
\end{align}
\end{lemma}

\subsection{Integration formulas}

We then introduce the integration formulas that will be crucial to perform the eigenvalue integrals discussed in this paper.

\begin{proposition}[Andréief's formula]
Let $(f_i(x))_{i \in \mathbb{N}}$ and $(g_i(x))_{i \in \mathbb{N}}$ be sequences of integrable functions.
Denoting $\dd{\mu(X)} = \prod_{i \in [N]} \dd{\mu(x_i)}$, the following identity holds,
\begin{align}
    \frac{1}{N!} \int \dd{\mu(X)} \det_{1 \le i, j \le N} f_i(x_j) \det_{1 \le i, j \le N} g_i(x_j) 
    & = \det_{1 \le i, j \le N} \rvev{f_i \mid g_j}
    \, .
    \label{eq:A_formula}
\end{align}
\end{proposition}

\begin{proposition}[de Bruijn's formula]
Let $N, M \in \mathbb{N}$, such that $N + M \in 2 \mathbb{N}$.
Let $A(x,y)$, $B(x,y)$ be skew-symmetric functions, and $S(x,y)$ generic integrable two-variable function.
Let $(f_{i}(x))_{i \in \mathbb{N}}$, $(g_{i}(x))_{i \in \mathbb{N}}$ be sequences of integrable functions as before.
Denoting $\dd{\mu(X)} = \prod_{i \in [N]} \dd{\mu(x_i)}$, etc, the following identities hold,
\begin{subequations}\label{eq:dB_formula}
\begin{align}
    & 
    \frac{1}{(2N)!} \int \dd{\mu(X)} \det_{i,j\in[2N]} f_j(x_i) \Pf_{i,j\in[2N]} A(x_i,x_j)
    = \Pf_{i,j\in[2N]} \rvev{ f_i \mid A \mid f_j}
    \, , \label{eq:dB_formula1}\\
    &
    \frac{1}{N!M!} \int \dd{\mu(X)} \dd{\mu(Y)} \det_{\substack{i\in[N],~ j \in [M] \\ k \in [N+M]}} \qty[f_k(x_i) \ g_k(y_j) ] \Pf_{\substack{i,j\in[N] \\ k,l \in[M]}} 
    \begin{bmatrix}
    A(x_i,{x}_j) & S(x_i,y_l) \\ - S^\text{T}(y_k,x_j) & B(y_k,y_l)
    \end{bmatrix}
    \nonumber \\
    & = \Pf_{i,j\in[N+M]} \qty[
    \rvev{f_i \mid A \mid f_j} + \rvev{f_i \mid S \mid g_j} - \rvev{g_i \mid S^\text{T} \mid f_j} + \rvev{g_i \mid B \mid g_j} ]
    \, . \label{eq:dB_formula2}
\end{align}
\end{subequations}
\end{proposition}

\begin{remark}
These formulas are due to C.~Andréief~\cite{Andreief:1886} and N.~G.~de Bruijn~\cite{deBruijn:1955JIMS}.
See \cite{Forrester:2019RMTA} for a historical survey on these formulas.
\end{remark}

\section{Matrix model with Pfaffian interaction}\label{sec:Pf}

In this Section, we introduce the generalized $\beta = 1$ and $\beta = 4$ matrix models.
We show that the corresponding partition functions are concisely described by the Pfaffian formula.

\subsection{$\beta = 1$ model}

We start to define the following model that we call the generalized $\beta = 1$ matrix model.
\begin{definition}[Generalized $\beta = 1$ matrix model]
Let $N \in 2 \mathbb{N}$, and $A(x,y)$ be a skew-symmetric function, $A(y,x) = - A(x,y)$.
Let $(f_{i-1}(x))_{i \in [N]}$ be a sequence of the integrable functions.
We denote the integral measure by $\dd{\mu(X)} = \prod_{i \in [N]} \dd{\mu(x_i)}$.
Then, we define the generalized $\beta = 1$ matrix model partition function as follows,
\begin{align}
    \mathcal{Z}_N^{(1)} = \frac{1}{N!} \int \dd{\mu(X)} \det_{i,j \in [N]} f_{j-1}(x_i) \Pf_{i,j \in [N]} A(x_i,x_j)
    \, .
    \label{eq:Z1}
\end{align}
\end{definition}
In fact, this generalized $\beta = 1$ model is reduced to the standard $\beta = 1$ model as follows.
\begin{proposition} \label{prop:beta1specialization}
Putting $A(x,y) = \operatorname{sgn}(x-y)$ and $f_{i-1}(x) = p_{N-i}(x)$, the generalized $\beta = 1$ model reduces to the standard $\beta = 1$ matrix model (real symmetric matrix model; orthogonal ensemble),
\begin{align}
    \mathcal{Z}_N^{(1)} \xrightarrow[f_{i-1} \ \to \ p_{N-i}]{A \ \to \ \operatorname{sgn}} \frac{1}{N!} \int \dd{\mu(X)} |\Delta_N(X)|
    \, .
\end{align}
\end{proposition}
\begin{proof}
Putting $A(x,y) = \operatorname{sgn}(x-y)$, we have
\begin{align}
    \Pf_{i,j\in[N]} \operatorname{sgn}(x_i - x_j) = \prod_{1 \le i < j \le N} \operatorname{sgn} (x_i - x_j)
    \, .
\end{align}
Putting $f_{i-1}(x) = p_{N-i}(x)$, the determinant factor becomes the Vandermonde determinant \eqref{eq:VDMdet}.
Hence, we have
\begin{align}
    \mathcal{Z}_N^{(1)} \ \longrightarrow \ \frac{1}{N!} \int \dd{\mu(X)} \Delta_N(X) \prod_{1 \le i < j \le N} \operatorname{sgn} (x_i - x_j)
    = \frac{1}{N!} \int \dd{\mu(X)} |\Delta_N(X)|
    \, .
\end{align}
This completes the proof.
\end{proof}

Then, the following Pfaffian formula holds for the generalized $\beta = 1$ model partition function.
\begin{proposition}
Let $N \in 2 \mathbb{N}$.
The following Pfaffian formula holds for the generalized $\beta = 1$ model partition function,
\begin{align}
    \mathcal{Z}_N^{(1)} = \Pf_{i,j\in[N]} \mathsf{N}^{(1)}_{i-1,j-1}
    \, ,
    \label{eq:Z1_pf}
\end{align}
where we define a skew-symmetric matrix,
\begin{align}
    \mathsf{N}^{(1)}_{i,j} = \rvev{ f_{i} \mid A \mid f_{j} }
    \, .
\end{align}
\end{proposition}
\begin{proof}
We can directly obtain this expression from de Bruijn's integration formula~\eqref{eq:dB_formula1}.
\end{proof}

\begin{remark}[Skew-orthogonal functions]
Defining the skew-orthogonal functions $(F_i)_{i \in [N]}$ obeying the relation for $i, j \in [N/2]$,
\begin{subequations}\label{eq:skeworth1}
\begin{align}
    \rvev{ F_{2i-2} \mid A \mid F_{2j-1} } = - \rvev{ F_{2j-1} \mid A \mid F_{2i-2} } & = h^{(1)}_i \delta_{i,j}
    \, , \\
    \rvev{ F_{2i-2} \mid A \mid F_{2j-2} } = \rvev{ F_{2i-1} \mid A \mid F_{2i-1} } & = 0
    \, ,
\end{align}
\end{subequations}
the partition function is given by
\begin{align}
    \mathcal{Z}_N^{(1)} = \prod_{i \in [N/2]} h_{i-1}^{(1)}
    \, .
\end{align}
\end{remark}

Similarly to the $\beta = 2$ matrix model, we may introduce the Christoffel--Darboux kernel for the $\beta = 1$ model as well, which will be a building block of the probability distribution function of the formal eigenvalues.
\begin{definition}[Christoffel--Darboux kernel]
Let $N \in 2 \mathbb{N}$.
We denote the inverse matrix of $\mathsf{N}^{(1)}$ by $\widetilde{\mathsf{N}}^{(1)}$.
Then, we define the Christoffel--Darboux (CD) kernel for the $\beta = 1$ model,
\begin{align}
    K_N^{(1)} (x,y) 
    & = \sum_{i,j = 0}^{N-1} f_{i}(x) \widetilde{\mathsf{N}}^{(1)}_{i,j} f_{j}(y)
    \nonumber \\
    & = \sum_{i = 0}^{N/2-1} \frac{F_{2i-1}(x) F_{2i-2}(y) - F_{2i-2}(x) F_{2i-1}(y)}{h_i^{(1)}}
    \, ,
\end{align}
where $(F_i)_{i \in [N]}$ is a set of the skew-orthogonal functions defined in \eqref{eq:skeworth1}.
The corresponding integral operator is given by
\begin{align}
    K_N^{(1)} & = \sum_{i,j=0}^{N-1} \ket{f_i} \widetilde{\mathsf{N}}^{(1)}_{i,j} \bra{f_j}
    \, .
\end{align}
\end{definition}
\begin{remark}
The CD kernel introduced here is slightly different from the quaternion kernel used in the standard analysis of the $\beta = 1$ model.
See, e.g., \cite{Mehta:2004RMT,Forrester:2010} for details.
\end{remark}
In fact, the CD kernel shows the property, called the self-reproducing property, which will be an important property to compute the correlation function.
\begin{proposition}
The CD kernel $K_N^{(1)}$ is self-reproducing,
\begin{subequations}
\begin{align}
 K_N^{(1)} \cdot A \cdot K_N^{(1)} & = K_N^{(1)}
 \, , \\
 \Tr \qty[ K_N^{(1)} \cdot A ] & = N
 \, .    
\end{align}
\end{subequations}
\end{proposition}
\begin{proof}
This property can be checked by direct calculation, 
\begin{align}
   K_N^{(1)} \cdot A \cdot K_N^{(1)} & = 
   \sum_{i,j,k,l = 0}^{N-1} \ket{f_i} \widetilde{\mathsf{N}}^{(1)}_{i,j} \underbrace{\rvev{ f_j \mid A \mid f_k }}_{\mathsf{N}^{(1)}_{j,k}} \widetilde{\mathsf{N}}^{(1)}_{k,l} \bra{f_l}
   = \sum_{i,l=0}^{N-1} \ket{f_i} \widetilde{\mathsf{N}}^{(1)}_{i,l} \bra{f_l}
   = K_N^{(1)}
   \, .
\end{align}
The trace condition is shown as follows,
\begin{align}
   \Tr\qty[ K_N^{(1)} \cdot A ] & 
   = \sum_{i,j=0}^{N-1} \int \int \dd{\mu(x)} \dd{\mu(y)} f_i(x) \widetilde{\mathsf{N}}^{(1)}_{i,j} f_j(y) A(y,x)
   = \sum_{i,j = 0}^{N-1} \widetilde{\mathsf{N}}^{(1)}_{i,j} \mathsf{N}^{(1)}_{j,i} = N
   \, .
\end{align}
\end{proof}
\begin{remark}
From the completeness condition, we obtain
\begin{align}
    K_\infty^{(1)}(x,y) = \lim_{N \to \infty} K_N^{(1)}(x,y) = \widetilde{A}(x,y)
    \, ,
    \label{eq:K1_infini}
\end{align}
where $\widetilde{A}(x,y)$ is the kernel of the inverse operator $A^{-1}$,
\begin{align}
    \int \dd{\mu(z)} A(x,z) \widetilde{A}(z,y) = \int \dd{\mu(z)} \widetilde{A}(x,z) A(z,y) = \delta(x-y)
    \, .
\end{align}
\end{remark}

Using this CD kernel, we may concisely express the probability distribution function as follows.
\begin{corollary}
The probability distribution function of the formal eigenvalues associated with the $\beta=1$ matrix model~\eqref{eq:Z1},
\begin{align}
    \mathsf{P}_N^{(1)}(X) \dd{\mu(X)} = \frac{1}{N! \mathcal{Z}_N^{(1)}} \det_{i,j \in [N]} f_{j-1}(x_i) \Pf_{i,j \in [N]} A(x_i,x_j) \dd{\mu(X)}
    \, ,
\end{align}
is written in terms of the CD kernel
\begin{align}
    \mathsf{P}_N^{(1)}(X) \dd{\mu(X)} & = \frac{1}{N!} \Pf_{i,j \in [N]} K_N^{(1)} (x_i,x_j) \Pf_{i,j \in [N]} A(x_i,x_j) \dd{\mu(X)}
    \, .
\end{align}
\end{corollary}
\begin{remark}
This structure of the probability distribution function (involving both the CD kernel and the interaction kernel) is similarly discussed in the context of the coupled matrix model~\cite{Eynard:1998JPA,Babinet:2022yij}.
\end{remark}

\subsection{$\beta = 4$ model}

We then consider the generalized version of the $\beta = 4$ matrix model.
\begin{definition}[Generalized $\beta = 4$ matrix model]
Let $N \in \mathbb{N}$, and $(f_{i-1}(x))_{i \in [2N]}$, $(g_{i-1}(x))_{i \in [2N]}$ be sequences of the integrable functions.
We define $2N$ variables $(z_i,z_{N+j}) = (x_i,y_j)$ for $i,j \in [N]$, and the corresponding integral measure by $\dd{\mu(Z)} = \prod_{i \in [2N]} \dd{\mu(z_i)}$.
For skew-symmetric functions, $A(x,y)$, $B(x,y)$, and a generic integrable function $S(x,y)$, we define the generalized $\beta = 4$ model partition function as follows,
\begin{align}
    \mathcal{Z}_{N}^{(4)} = \frac{1}{N!^2} \int \dd{\mu(Z)} \det_{\substack{i \in [N], j \in [2N]}} \qty[ f_{j-1}(x_i) \ g_{j-1}(y_i) ] \Pf_{\substack{i,j \in [N]}} 
    \mathsf{A}(x_i,x_j;y_i,y_j)
    \, ,
    \label{eq:Z4}
\end{align}
where we denote 
\begin{align}
    \mathsf{A}(x,\tilde{x};y,\tilde{y}) = 
    \begin{bmatrix}
    A(x,\tilde{x}) & S(x,\tilde{y}) \\ - S^\text{T}(y,\tilde{x}) & B(y,\tilde{y})
    \end{bmatrix}
    \, .
\end{align}
\end{definition}
As is the case of the $\beta = 1$ model, the generalized $\beta = 4$ model is reduced to the standard $\beta = 4$ model via the specialization as follows.
\begin{lemma}\label{lemma:beta4specialization}
Putting
\begin{gather}
    f_{i-1}(x) = p_{2N-i}(x) \, , \quad
    g_{i-1}(x) = p_{2N-i}'(x) \, , \quad 
    A(x,y) = B(x,y) = 0 \, , \quad
    S(x,y) = \delta(x-y) \, ,
\end{gather}
the generalized $\beta = 4$ model reduces to the standard $\beta = 4$ matrix model (quaternion self-dual; symplectic ensemble),
\begin{align}
    \mathcal{Z}_{N,N}^{(4)} \ \longrightarrow \ \frac{1}{N!} \int \dd{\mu(X)} \Delta_N(X)^4
    \, .
\end{align}
\end{lemma}
\begin{proof}
Applying the formula~\eqref{eq:Pf2det} and Andréief's formula~\eqref{eq:A_formula}, we obtain
\begin{align}
    \mathcal{Z}_{N,N}^{(4)} & = \frac{1}{N!^2} \int \dd{\mu(Z)} 
    \det_{\substack{i \in [N] \\ j \in [2N]}} \qty[ p_{2N-j}(x_i) \ p_{2N-j}'(y_i) ] \Pf_{\substack{i,j \in [N]}} 
    \begin{bmatrix}
    0 & \delta(x_i-y_j) \\ - \delta(y_i-x_j) & 0
    \end{bmatrix}
    \nonumber \\
    & = \frac{1}{N!^2} \int \dd{\mu(Z)} 
    \det_{\substack{i \in [N] \\ j \in [2N]}} \qty[ p_{2N-j}(x_i) \ p_{2N-j}'(y_i) ] \det_{\substack{i,j \in [N]}} \delta(x_i-y_j) 
    \nonumber \\
    & = \frac{1}{N!} \int \dd{\mu(X)} \det_{\substack{i \in [N] \\ j \in [2N]}} \qty[ p_{2N-j}(x_i) \ p_{2N-j}'(x_i) ] 
    \nonumber \\
    & = \frac{1}{N!} \int \dd{\mu(X)} \Delta_N(X)^4
    \, .
\end{align}
This completes the proof.
\end{proof}

Then, the following Pfaffian formula holds for the generalized $\beta = 4$ model partition function.
\begin{proposition}
The following Pfaffian formula holds for the generalized $\beta = 4$ model partition function,
\begin{align}
    \mathcal{Z}_{N}^{(4)} = \Pf_{i,j\in[2N]} \mathsf{N}^{(4)}_{i-1,j-1}
    \, ,
    \label{eq:Z4_pf}
\end{align}
where we define a skew-symmetric matrix,
\begin{align}
    \mathsf{N}^{(4)}_{i,j} & = \rvev{f_{i} \mid A \mid f_{j} } + \rvev{f_{i} \mid S \mid g_{j}} - \rvev{g_{i} \mid S^\text{T} \mid f_{j}} + \rvev{g_{i} \mid B \mid g_{j}} 
    \nonumber \\
    & = 
    \begin{bmatrix}
    \bra{f_i} & \bra{g_i}
    \end{bmatrix}
    \begin{bmatrix}
    A & S \\ - S^\text{T} & B
    \end{bmatrix}
    \begin{bmatrix}
    \ket{f_j} \\ \ket{g_j}
    \end{bmatrix}
    \nonumber \\
    & =: \rvev{f_i,g_i \mid \mathsf{A} \mid f_j,g_j}
    \, .
\end{align}
\end{proposition}
\begin{proof}
We obtain this formula by using de Bruijn's integration formula~\eqref{eq:dB_formula2}.
\end{proof}

\begin{remark}[Skew-biorthogonal functions]
Defining the skew-biorthogonal functions $(F_i,G_i)_{i \in [N]}$ obeying the relation for $i, j \in [N]$,
\begin{subequations}\label{eq:skeworth4}
\begin{align}
    \rvev{ F_{2i-2}, G_{2i-2} \mid \mathsf{A} \mid F_{2j-1}, G_{2j-1} } = - \rvev{ F_{2j-1}, G_{2j-1} \mid \mathsf{A} \mid F_{2i-2}, G_{2i-2} } & = h^{(4)}_i \delta_{i,j}
    \, , \\
    \rvev{ F_{2i-2}, G_{2i-2} \mid \mathsf{A} \mid F_{2j-2}, G_{2j-2} } = \rvev{ F_{2i-1}, G_{2i-1} \mid \mathsf{A} \mid F_{2j-1}, G_{2j-1} }
    & = 0
    \, ,
\end{align}
\end{subequations}
the partition function is given by
\begin{align}
    \mathcal{Z}_N^{(4)} = \prod_{i \in [N]} h_{i-1}^{(4)}
    \, .
\end{align}
\end{remark}

Similarly to the $\beta = 1$ model, we consider the CD kernel for the $\beta = 4$ model, which is a matrix analog of the ordinary kernel.
\begin{definition}[Matrix Christoffel--Darboux kernel]
Let $N \in \mathbb{N}$.
We denote the inverse matrix of $\mathsf{N}^{(4)}$ by $\widetilde{\mathsf{N}}^{(4)}$.
Then, we define the matrix Christoffel--Darboux (CD) kernel for the $\beta = 4$ model,
\begin{align}
    & \mathsf{K}_{N}^{(4)} (x,\tilde{x}; y, \tilde{y})
    \nonumber \\
    & = 
    \sum_{i,j =0}^{2N-1}
    \begin{bmatrix}
    f_i(x) \\ g_i(y)
    \end{bmatrix}
    \widetilde{\mathsf{N}}^{(4)}_{i,j}
    \begin{bmatrix}
    f_j(\tilde{x}) & g_j(\tilde{y})
    \end{bmatrix}
    \nonumber \\
    & = \sum_{i=0}^{N-1} \frac{1}{h_i^{(4)}}
    \begin{bmatrix}
    F_{2i+1}(x) F_{2i}(\tilde{x}) - F_{2i}(x) F_{2i+1}(\tilde{x}) & 
    F_{2i+1}(x) G_{2i}(\tilde{y}) - F_{2i}(x) G_{2i+1}(\tilde{y}) \\
    G_{2i+1}(y) F_{2i}(\tilde{x}) - G_{2i}(y) F_{2i+1}(\tilde{x}) &
    G_{2i+1}(y) G_{2i}(\tilde{y}) - G_{2i}(y) G_{2i+1}(\tilde{y})
    \end{bmatrix}
    \, ,
\end{align}
where $(F_{i-1},G_{i-1})_{i \in [2N]}$ is a set of the skew-biorthogonal functions defined in \eqref{eq:skeworth4}.
In the operator formalism, it is given by
\begin{align}
    \mathsf{K}_N^{(4)} = \sum_{i,j=0}^{2N-1} 
    \begin{bmatrix}
    \ket{f_i} \\ \ket{g_i}
    \end{bmatrix}
    \widetilde{\mathsf{N}}^{(4)}_{i,j}
    \begin{bmatrix}
    \bra{f_j} & \bra{g_j}
    \end{bmatrix}
    =
    \sum_{i,j=0}^{2N-1} 
    \begin{bmatrix}
    \ket{f_i} \widetilde{\mathsf{N}}^{(4)}_{i,j} \bra{f_j} &
    \ket{f_i} \widetilde{\mathsf{N}}^{(4)}_{i,j} \bra{g_j} \\
    \ket{g_i} \widetilde{\mathsf{N}}^{(4)}_{i,j} \bra{f_j} &
    \ket{g_i} \widetilde{\mathsf{N}}^{(4)}_{i,j} \bra{g_j}
    \end{bmatrix}
    \, .
\end{align}
\end{definition}
This matrix CD kernel similarly exhibits the self-reproducing property.
\begin{proposition}
The CD kernel $K_N^{(4)}$ is self-reproducing,
\begin{subequations}
\begin{align}
 \mathsf{K}_N^{(4)} \cdot \mathsf{A} \cdot \mathsf{K}_N^{(4)} & = \mathsf{K}_N^{(4)}
 \, , \\
 \Tr \qty[ \mathsf{K}_N^{(4)} \cdot \mathsf{A} ] & = 2N
 \, .    
\end{align}
\end{subequations}
\end{proposition}
\begin{proof}
This property can be checked by direct calculation as follows,
\begin{align}
   \mathsf{K}_N^{(4)} \cdot \mathsf{A} \cdot \mathsf{K}_N^{(4)} & = 
   \sum_{i,j,k,l = 0}^{2N-1} 
   \begin{bmatrix}
    \ket{f_i} \\ \ket{g_i}
    \end{bmatrix}
    \widetilde{\mathsf{N}}^{(4)}_{i,j}
    \begin{bmatrix}
    \bra{f_j} & \bra{g_j}
    \end{bmatrix}
    \mathsf{A}
    \begin{bmatrix}
    \ket{f_k} \\ \ket{g_k}
    \end{bmatrix}
    \widetilde{\mathsf{N}}^{(4)}_{k,l}
    \begin{bmatrix}
    \bra{f_l} & \bra{g_l}
    \end{bmatrix}
    \nonumber \\
    & = 
    \sum_{i,j,k,l = 0}^{2N-1}
   \begin{bmatrix}
    \ket{f_i} \\ \ket{g_i}
    \end{bmatrix}
    \widetilde{\mathsf{N}}^{(4)}_{i,j}
    \mathsf{N}^{(4)}_{j,k}
    \widetilde{\mathsf{N}}^{(4)}_{k,l}
    \begin{bmatrix}
    \bra{f_l} & \bra{g_l}
    \end{bmatrix}
    = \mathsf{K}_N^{(4)}
   \, .
\end{align}
The trace condition is shown as follows,
\begin{align}
   \Tr\qty[ \mathsf{K}_N^{(4)} \cdot \mathsf{A} ]
   & = \sum_{i,j = 0}^{2N-1} \int \dd{\mu(x)} \dd{\mu(\tilde{x})} \dd{\mu(y)} \dd{\mu(\tilde{y})}
   \Tr \qty(
    \begin{bmatrix}
    f_i(x) \\ g_i(y)
    \end{bmatrix}
    \widetilde{N}^{(4)}_{i,j}
    \begin{bmatrix}
    f_j(\tilde{x}) & g_j(\tilde{y})
    \end{bmatrix}
    \begin{bmatrix}
    A(x,\tilde{x}) & S(x,\tilde{y}) \\ - S(y,\tilde{x}) & B(y,\tilde{y})
    \end{bmatrix}
    )
   \nonumber \\
   & = \sum_{i,j = 0}^{2N-1} \widetilde{\mathsf{N}}^{(4)}_{i,j} \mathsf{N}^{(4)}_{j,i} = 2N
   \, .
\end{align}
This completes the proof.
\end{proof}
\begin{remark}
From the completeness condition, we obtain
\begin{align}
    \mathsf{K}_\infty^{(4)} = \lim_{N \to \infty} \mathsf{K}_N^{(4)} = \widetilde{\mathsf{A}}
    \, ,
    \label{eq:K4_infini}
\end{align}
which is the inverse of the integral operator $\mathsf{A}$,
\begin{align}
    \mathsf{A} \cdot \widetilde{\mathsf{A}} = \widetilde{\mathsf{A}} \cdot \mathsf{A} = 1
    \, .
\end{align}
\end{remark}

Using the matrix CD kernel, the probability distribution function of the formal eigenvalues may be written as follows.
\begin{corollary}
The probability distribution function of the formal eigenvalues associated with
the $\beta = 4$ matrix model \eqref{eq:Z4},
\begin{align}
    &
    \mathsf{P}_{N}^{(4)}(X;Y) \dd{\mu(X;Y)} 
    \nonumber \\
    & = \frac{1}{N!^2 \mathcal{Z}_N^{(4)}} \det_{\substack{i \in [N], \ j \in [2N] }} \qty[ f_{j-1}(x_i) \ g_{j-1}(y_i) ] \Pf_{\substack{i,j \in [N]}} 
    \begin{bmatrix}
    A(x_i,x_j) & S(x_i,y_j) \\ - S^\text{T}(y_i,x_j) & B(y_i,y_j)
    \end{bmatrix}
    \dd{\mu(X;Y)}
\end{align}
is written in terms of the CD kernel,
\begin{align}
    \mathsf{P}_{N}^{(4)}(X;Y) \dd{\mu(X;Y)} & = \frac{1}{N!^2} \Pf_{i,j \in [N]} \mathsf{K}_{N}^{(4)} (x_i,x_j;y_i,y_j) \Pf_{\substack{i,j \in [N]}} 
    \begin{bmatrix}
    A(x_i,x_j) & S(x_i,y_j) \\ - S^\text{T}(y_i,x_j) & B(y_i,y_j)
    \end{bmatrix}
    \dd{\mu(X;Y)}
    \, .
\end{align}
\end{corollary}

\section{Characteristic polynomials}\label{sec:ch_poly}

In this Section, we consider the characteristic polynomial averages of the $\beta = 1$ and $\beta = 4$ matrix models.
We first define the expectation value with respect to the corresponding probability distribution functions.
\begin{definition}[Expectation value]
We define the expectation value with respect to the
probability distribution function $\mathsf{P}_N^{(\beta)}(X)$ for $\beta = 1, 4$ as follows,
\begin{align}
    \rvev{ \mathcal{O}(X) }_{\beta} & = \int \dd{\mu(X)} \mathsf{P}_N^{(\beta)}(X) \, \mathcal{O}(X)
    \, .
\end{align}
We remark that the number of variables are $N$ and $2N$ for $\beta = 1$ and $\beta = 4$, respectively.
\end{definition}

\subsection{$\beta = 1$ model}

We consider the following special case of the $\beta = 1$ model~\eqref{eq:Z1} for $N \in 2\mathbb{N}$ with a skew-symmetric function $A(x,y) = - A(y,x)$:
\begin{align}
    \mathcal{Z}_N^{(1)} = \frac{1}{N!} \int \dd{\mu(X)} \Delta_N(X) \Pf_{i,j \in [N]} A(x_i,x_j)
    \, .
\end{align}
Then, we study the Schur polynomial average with respect to the $\beta = 1$ model, which would be a building block of the correlation functions discussed below.
\begin{lemma}[Schur polynomial average]
Let $N \in 2 \mathbb{N}$.
The Schur polynomial average of the $\beta = 1$ model is given by a rank $N$ Pfaffian as follows,
\begin{align}
    \expval{s_\lambda(X)}_1 = \frac{1}{\mathcal{Z}_N^{(1)}} \Pf_{i,j \in [N]} \qty[ \rvev{ x^{\lambda_i + N - i} \mid A \mid x^{\lambda_j + N - j} } ]
    \, .
\end{align}
\end{lemma}
\begin{proof}
From the definition of the Schur polynomial~\eqref{eq:Schur_def}, we have
\begin{align}
    \expval{s_\lambda(X)}_1 = \frac{1}{N! \mathcal{Z}_N^{(1)}} \int \dd{\mu(X)} \det_{i,j\in [N]} x_i^{\lambda_j + N - j} \Pf_{i,j\in [N]} A(x_i,x_j)
    \, .
\end{align}
Then, applying the de Bruijn's formula~\eqref{eq:dB_formula1}, we arrive at the result.
\end{proof}

Applying this Schur polynomial average together with the Schur polynomial expansion of the characteristic polynomial, we obtain the Pfaffian formula for the characteristic polynomial average.
\begin{proposition}[Characteristic polynomial]
Let $N,M \in 2 \mathbb{N}$.
For $Z = (z_\alpha)_{i \in [M]}$, the characteristic polynomial average of the $\beta = 1$ model is given by the rank $M$ Pfaffian of the degree-$(N+M)$ CD kernel,
\begin{align}
    \expval{\prod_{\alpha \in [M]} \det(z_\alpha - X)}_1
    & = \frac{1}{\Delta_M(Z)} \frac{\mathcal{Z}_{N+M}^{(1)}}{\mathcal{Z}_N^{(1)}} \Pf_{\alpha,\beta \in [M]} K_{N+M}^{(1)} (z_\alpha,z_\beta)  
    \, .
\end{align}
\end{proposition}
\begin{proof}
We apply the Schur polynomial expansion formula (Lemma~\ref{lemma:Schur_expansion}) to obtain
\begin{align}
    & \expval{\prod_{\alpha \in [M]} \det(z_\alpha - X)}_1
    \nonumber \\
    & = \sum_{\lambda \subseteq (M^N)} (-1)^{|\lambda|} s_{\lambda^\vee}(Z) \expval{s_\lambda(X)}_1
    \nonumber \\
    & = \frac{1}{\mathcal{Z}_N^{(1)} \Delta_M(Z)} \sum_{\lambda \subseteq (M^N)} (-1)^{|\lambda|} \det_{\alpha,\beta \in [M]} z_\alpha^{\lambda^\vee_\beta + M - \beta} \Pf_{\substack{i,j \in [N] }} \qty[ \rvev{ x^{\lambda_i + N - i} \mid A \mid x^{\lambda_j + N - j} } ]
    \, .
\end{align}
Applying the Laplace-type expansion (Lemma~\ref{lemma:Laplace_exp}) and the Pfaffian formula for the block matrix~\eqref{eq:Pf_block}, we have
\begin{align}
    \expval{\prod_{\alpha \in [M]} \det(z_\alpha - X)}_1
    & = \frac{1}{\mathcal{Z}_N^{(1)} \Delta_M(Z)} \Pf_{\substack{i,j \in [N+M] \\ \alpha,\beta \in [M]}}
    \begin{bmatrix}
    \rvev{p_{N+M-i} \mid A \mid p_{N+M-j}} & p_{N+M-i}(z_\alpha) \\
    - p_{N+M-i}(z_{\beta}) & 0
    \end{bmatrix}
    \nonumber \\
    & = \frac{\mathcal{Z}_{N+M}^{(1)}}{\mathcal{Z}_N^{(1)} \Delta_M(Z)}
    \Pf_{\alpha,\beta \in [M]} \qty[ \sum_{i,j=0}^{N+M-1} p_i(z_\alpha) \widetilde{N}^{(1)}_{i,j} p_j(z_\beta) ]
    \nonumber \\
    & = \frac{1}{\Delta_M(Z)} \frac{\mathcal{Z}_{N+M}^{(1)}}{\mathcal{Z}_N^{(1)}} \Pf_{\alpha,\beta \in [M]} K_{N+M}^{(1)} (z_\alpha,z_\beta)
    \, .
\end{align}
This completes the proof.
\end{proof}

Next, we consider the average of the characteristic polynomial inverse.
For this purpose, we define the dual CD kernel as follows.
\begin{definition}[Dual Christoffel--Darboux kernel]
Let $N \in 2 \mathbb{N}$.
We define the integral operator as follows,
\begin{align}
    \bar{A}_N 
    = A - \sum_{k,l=0}^{N-1} A \ket{p_k} \widetilde{N}^{(1)}_{k,l} \bra{p_l} A
    = A \qty( \widetilde{A}  - K_{N}^{(1)} ) A
    \, .
    \label{eq:Abar}
\end{align}
Recalling the relation~\eqref{eq:K1_infini}, we may write
\begin{align}
    \bar{A}_N = \sum_{i=N/2}^\infty A \qty( \ket{\psi_{2i-1}} \bra{\psi_{2i-2}} - \ket{\psi_{2i-2}} \bra{\psi_{2i-1}} ) A
    \, ,
\end{align}
where $(\psi_k)_{k \in \mathbb{N}}$ is a set of degree-$k$ skew-orthonormal polynomials, $\rvev{ \psi_{2i-2} \mid A \mid \psi_{2j-1} } = \delta_{i,j}$, $\rvev{ \psi_{2i} \mid A \mid \psi_{2j}} = \rvev{\psi_{2i-1} \mid A \mid \psi_{2j-1}} = 0$.
Then, we define the dual Christoffel--Darboux kernel through the double Hilbert transform,
\begin{align}
    \widetilde{K}^{(1)}_{N}(z,w) 
    & = 
    \int \dd{\mu(x)} \dd{\mu(\tilde{x})} \frac{ \bar{A}_N(x,\tilde{x}) }{(z - x)(w - \tilde{x})}
    \nonumber \\
    & =
    \int \dd{\mu(x)} \dd{\mu(\tilde{x})} \sum_{i= N/2}^\infty  \frac{ \bra{x \mid A } \qty( \ket{\psi_{2i-1}} \bra{\psi_{2i-2}} -  \ket{\psi_{2i-2}} \bra{\psi_{2i-1}} ) \ket{ A \mid \tilde{x} } }{(z - x)(w - \tilde{x})}
    \, .
\end{align}
\end{definition}

Using this dual CD kernel, we obtain the following Pfaffian formula for the characteristic polynomial inverse.
\begin{proposition}[Characteristic polynomial inverse]
Let $N,M \in 2 \mathbb{N}$.
For $Z = (z_\alpha)_{\alpha \in [M]}$, the characteristic polynomial inverse average of the $\beta = 1$ model is given by a Pfaffian of the dual CD kernel,
\begin{subequations}
\begin{align}
    & N \ge M : &
    \expval{\prod_{\alpha \in [M]} \det(z_\alpha - X)^{-1}}_1
    & = (-1)^{NM} \frac{\mathcal{Z}_{N-M}^{(1)}/\mathcal{Z}_N^{(1)}}{\Delta_M(Z)} \Pf_{\alpha,\beta \in [M]} \widetilde{K}_{N-M}^{(1)} (z_\alpha,z_\beta)  
    \, , \\
    & N \le M : &
    \expval{\prod_{\alpha \in [M]} \det(z_\alpha - X)^{-1}}_1
    & = \frac{1/\mathcal{Z}_N^{(1)}}{\Delta_M(Z)}
    \Pf_{\alpha,\beta \in [M]} \widetilde{K}_0^{(1)}(z_\alpha,z_\beta)
    \nonumber \\
    &&& \quad \times 
    \Pf_{\substack{k,l \in [M-N]}}
    \qty[ \sum_{\alpha,\beta=1}^M 
    \rvev{p_{M-N-k} \mid z_\alpha} \rvev{z_\alpha \mid \frac{1}{\widetilde{K}^{(1)}_0} \mid z_\beta} \rvev{z_\beta \mid p_{M-N-l}} ]
    \, .
\end{align}
\end{subequations}
\end{proposition}
\begin{proof}
We first consider the case $N \ge M$.
In this case, we apply the Schur polynomial expansion formula (Lemma~\ref{lemma:Schur_expansion}) to obtain
\begin{align}
    & \expval{\prod_{\alpha \in [M]} \det(z_\alpha - X)^{-1}}_1
    = (-1)^{NM} \sum_{\ell(\lambda) \le M} s_{\lambda}(Z) \expval{\frac{s_\lambda(X^{-1})}{\det X^{M}}}_1
    \nonumber \\
    & = \frac{(-1)^{NM}}{\mathcal{Z}_N^{(1)} \Delta_M(Z)} \sum_{\ell(\lambda) \le M} \det_{\alpha,\beta \in [M]} z_\alpha^{\lambda_\beta + M - \beta} 
    \nonumber \\
    & \hspace{5em} \times
    \Pf_{\substack{\alpha,\beta \in [M] \\ k,l \in [N-M] }} 
    \begin{bmatrix}
    \rvev{ x^{- \lambda_\alpha - M + \alpha - 1} \mid A \mid x^{- \lambda_\beta - M + \beta - 1} } & \rvev{ x^{- \lambda_\alpha - M + \alpha - 1} \mid A \mid x^{N-M-l} } \\
    \rvev{ x^{N-M-k} \mid A \mid x^{- \lambda_\beta - M + \beta - 1} } & \rvev{ x^{N-M-k} \mid A \mid x^{N-M-l} }
    \end{bmatrix}
    \nonumber \\
    & = \frac{(-1)^{NM}}{\Delta_M(Z)} \frac{\mathcal{Z}_{N-M}^{(1)}}{\mathcal{Z}_N^{(1)}} \sum_{\ell(\lambda) \le M} \det_{\alpha,\beta \in [M]} z_\alpha^{\lambda_\beta + M - \beta}
    \Pf_{\substack{\alpha,\beta \in [M] }} 
    \qty[
    \rvev{ x^{- \lambda_\alpha - M + \alpha - 1} \mid \bar{A}_{N-M} \mid x^{- \lambda_\beta - M + \beta - 1} } 
    ]
    \, .
\end{align}
Applying the Cauchy--Binet-type expansion (Lemma~\ref{lemma:CB_exp}), we have
\begin{align}
    & \sum_{\ell(\lambda) \le M} \det_{\alpha,\beta \in [M]} z_\alpha^{\lambda_\beta + M - \beta} \Pf_{\substack{\alpha,\beta \in [M] }} 
    \begin{bmatrix}
    \rvev{ x^{- \lambda_\alpha - M + \alpha - 1} \mid \bar{A}_{N-M} \mid x^{- \lambda_\beta - M + \beta - 1} }
    \end{bmatrix}
    \nonumber \\
    & = \frac{1}{M!} \sum_{\substack{r_\alpha = 0, \ldots,\infty \\ r_\alpha \neq r_\beta}} 
    \det_{\alpha,\beta \in [M]} z_\alpha^{r_\beta} \Pf_{\substack{\alpha,\beta \in [M] }} 
    \begin{bmatrix}
    \rvev{ x^{-r_\alpha-1} \mid \bar{A}_{N-M} \mid x^{-r_\beta-1} }
    \end{bmatrix}
    \nonumber \\
    & = \Pf_{\alpha,\beta \in [M]}
    \qty[ \sum_{r,s=0}^\infty z_\alpha^{r} z_\beta^{s} \rvev{ x^{-r-1} \mid \bar{A}_{N-M} \mid x^{-s-1} } ]
    = \Pf_{\alpha,\beta \in [M]}
    \qty[ \widetilde{K}_{N-M}^{(1)} (z_\alpha,z_\beta) ]
    \, .
\end{align}
The last expression is obtained as follows,
\begin{align}
    \sum_{r,s=0}^\infty z^{r} w^{s} \rvev{ x^{-r-1} \mid \bar{A}_{N-M} \mid x^{-s-1} }
    & = 
    \int \dd{\mu(x)} \dd{\mu(\tilde{x})}  \frac{\bar{A}_{N-M}(x,\tilde{x})}{(z - x)(w - \tilde{x})}
    \nonumber \\
    & = \widetilde{K}^{(1)}_{N-M}(z,w)
    \, .
    \label{eq:sum_to_Hilbert}
\end{align}
This completes the derivation for the case $N \ge M$.
For the case $N \le M$, on the other hand, the Schur polynomial expansion yields
\begin{align}
    & \expval{\prod_{\alpha \in [M]} \det(z_\alpha - X)^{-1}}_1
    = \sum_{\ell(\lambda) \le N} \frac{ s_{\lambda}(Z^{-1}) }{\det Z^N} \expval{s_\lambda(X)}_1
    \nonumber \\
    & = \frac{1}{\mathcal{Z}_N^{(1)} \Delta_M(Z)} \sum_{\ell(\lambda) \le N} \det_{\substack{\alpha \in [M], \beta \in [N] \\ k \in [M-N]}} 
    \begin{bmatrix}
    z_\alpha^{- \lambda_\beta - N + \beta - 1} \\ p_{M - N - k}(z_\alpha)
    \end{bmatrix}
    \Pf_{\alpha,\beta \in [N]}
    \rvev{ x^{\lambda_\alpha + N - \alpha} \mid A \mid x^{\lambda_\beta + N - \beta} }
    \, .
\end{align}
Applying the Cauchy--Binet-type expansion~\eqref{eq:CB0}, we obtain
\begin{align}
    & \expval{\prod_{\alpha \in [M]} \det(z_\alpha - X)^{-1}}_1
    \nonumber \\
    & = \frac{1/\mathcal{Z}_N^{(1)}}{\Delta_M(Z)}
    \Pf_{\substack{ \alpha,\beta \in [M] \\ k,l \in [M-N]}}
    \begin{bmatrix}
    \widetilde{K}_0^{(1)}(z_\alpha,z_\beta) & p_{M-N-l}(z_\alpha) \\
    -p_{M-N-k}(z_\beta) & 0
    \end{bmatrix}
    \nonumber \\
    & = \frac{1/\mathcal{Z}_N^{(1)}}{\Delta_M(Z)}
    \Pf_{\alpha,\beta \in [M]} \widetilde{K}_0^{(1)}(z_\alpha,z_\beta)
    \Pf_{\substack{k,l \in [M-N]}}
    \qty[ \sum_{\alpha,\beta=1}^M 
    \rvev{p_{M-N-k} \mid z_\alpha} \rvev{z_\alpha \mid \frac{1}{\widetilde{K}^{(1)}_0} \mid z_\beta} \rvev{z_\beta \mid p_{M-N-l}} ]
    \, .
\end{align}
This completes the proof.
\end{proof}

\begin{remark}
Applying the characteristic polynomial formulas for the two-point function $(M = 2)$, we obtain an alternative form of the CD kernels,
\begin{subequations}
\begin{align}
   \frac{\mathcal{Z}_{N+2}^{(1)}}{\mathcal{Z}_N^{(1)}} K_{N+2}^{(1)}(z,w) & = (z-w) \expval{\det(z-X) \det(w-X)}_1
   \, , \\
   \frac{\mathcal{Z}_{N-2}^{(1)}}{\mathcal{Z}_N^{(1)}} \widetilde{K}_{N-2}^{(1)}(z,w) & = (z-w) \expval{\det(z-X)^{-1} \det(w-X)^{-1}}_1
   \, ,
\end{align}
\end{subequations}
which was originally observed in~\cite{Borodin:2006CPAM}.
Such a relation between the characteristic polynomial average and the CD kernel is also available for the $\beta = 2$ model (a.k.a., generalized Heine formula).
See, e.g.~\cite{Eynard:2015aea} for details.
\end{remark}

\subsection{$\beta = 4$ model}

We consider the following special case of the $\beta = 4$ model~\eqref{eq:Z4} for $N \in \mathbb{N}$ and skew-symmetric functions $A(x,y) = - A(y,x)$ and $B(x,y) = - B(y,x)$, and a generic integrable function $S(x,y)$:
\begin{align}
    \mathcal{Z}_N^{(4)} & = \frac{1}{N!^2} \int \dd{\mu{(X;Y)}} \det_{i \in [N], j \in [2N]} 
    \begin{bmatrix}
    p_{2N-j}(x_i) & p_{2N-j}'(y_i)
    \end{bmatrix}
    \Pf_{\substack{i,j \in [N]}} 
    \begin{bmatrix}
    A(x_i,x_j) & S(x_i,y_j) \\ - S^\text{T}(y_i,x_j) & B(y_i,y_j)
    \end{bmatrix}
    \nonumber \\
    & = \frac{1}{N!^2} \int \dd{\mu{(X;Y)}} \det_{i \in [N], j \in [2N]} 
    \begin{bmatrix}
    p_{2N-j}(x_i) & p_{2N-j}'(y_i)
    \end{bmatrix}
    \Pf_{\substack{i,j \in [N]}} 
    \mathsf{A}(x_i,x_j;y_i,y_j)
    \, .
\end{align}

In order to consider the characteristic polynomial average for the $\beta = 4$ model, we compute the average of the modified Schur polynomial as follows.
\begin{lemma}[Modified Schur polynomial average]
Let $N \in \mathbb{N}$.
Then, the modified Schur polynomial average of the $\beta = 4$ model is given by a rank $2N$ Pfaffian as follows,
\begin{align}
    \expval{\mathsf{s}_\lambda(X;Y)}_4 = \frac{1}{\mathcal{Z}_{N}^{(4)}} \Pf_{i,j \in [2N]} 
    \rvev{ x^{\lambda_i + 2N - i} \mid \mathsf{A} \mid y^{\lambda_j + 2N - j} }
    \, ,
\end{align}
where we here denote the norm by
\begin{align}
    \rvev{ f \mid \mathsf{A} \mid g } := \rvev{ f, g' \mid \mathsf{A} \mid f, g' }
    \, .
\end{align}
\end{lemma}
\begin{proof}
From the definition of the modified Schur polynomial~\eqref{eq:Schur_def_mod}, we have
\begin{align}
    & \expval{\mathsf{s}_\lambda(X;Y)}_4 
    \nonumber \\
    & = \frac{1}{N!^2 \mathcal{Z}_N^{(4)}} \int \dd{\mu(X; Y)} \det_{i \in [N], j \in [2N]} 
    \begin{bmatrix}
    x_i^{\lambda_j + 2N - j} & (y_i^{\lambda_j + 2N - j})'
    \end{bmatrix}
    \Pf_{\substack{i,j \in [N]}} 
    \begin{bmatrix}
    A(x_i,x_j) & S(x_i,y_j) \\ - S^\text{T}(y_i,x_j) & B(y_i,y_j)
    \end{bmatrix}
    \, .
\end{align}
Then, applying the de Bruijn's formula~\eqref{eq:dB_formula2}, we arrive at the formula.
\end{proof}

\begin{remark}
Putting 
\begin{align}
 A(x,y) = B(x,y) = 0 \, , \qquad S(x,y) = \delta(x-y)
 \, ,
 \label{eq:beta4specialization}
\end{align}
as in Lemma~\ref{lemma:beta4specialization}, the average of the modified Schur polynomial is reduced to that for the ordinary $2N$-variable Schur polynomial with respect to the symplectic ensemble,
\begin{align}
    \expval{\mathsf{s}_\lambda(X;Y)}_4 \ \xrightarrow{\eqref{eq:beta4specialization}} \ \expval{{s}_\lambda(X;X)}_4
    \, .
\end{align}
\end{remark}

In order to describe the characteristic polynomial average, we define the auxiliary kernel, which is the $(1,1)$-component of the matrix kernel $\mathsf{K}_N^{(4)}$, by
\begin{align}
    K_N^{(4)}(x,y) = \sum_{i,j=0}^{2N-1} p_i(x) \widetilde{\mathsf{N}}^{(4)}_{i,j} p_j(y)
    \, .
\end{align}
Then, we may write the matrix kernel in terms of the auxiliary kernel as follows,
\begin{align}
    \mathsf{K}_N^{(4)}(x,\tilde{x};y,\tilde{y}) = 
    \begin{bmatrix}
    K_N^{(4)}(x,\tilde{x}) & \partial_{\tilde{y}} K_N^{(4)}(x,\tilde{y}) \\
    \partial_{y} K_N^{(4)}({y},\tilde{x}) & \partial^2_{y,\tilde{y}} K_N^{(4)}(y,\tilde{y})
    \end{bmatrix}
    \, .
\end{align}
With this auxiliary kernel, we obtain the following Pfaffian formula for the characteristic polynomial average.
\begin{proposition}
Let $N \in \mathbb{N}$, $M \in 2 \mathbb{N}$, and we denote $Z = (z_\alpha)_{\alpha \in [M]}$.
Then, the following Pfaffian formula holds,
\begin{align}
    \sum_{\lambda \subseteq (M^{2N})} (-1)^{|\lambda|} s_\lambda(Z) \expval{\mathsf{s}_\lambda(X;Y)}_4
    & = \frac{1}{\Delta_M(Z)} \frac{\mathcal{Z}_{N+M/2}^{(4)}}{\mathcal{Z}_N^{(4)}}
    \Pf_{\alpha,\beta \in[M]} K_{N+M/2}^{(4)} (z_\alpha,z_\beta)
    \, .
\end{align}
\end{proposition}
\begin{proof}
Applying the Laplace-type expansion (Lemma~\ref{lemma:Laplace_exp}) and the formula for the block matrix~\eqref{eq:Pf_block}, we have
\begin{align}
    & \sum_{\lambda \subseteq (M^{2N})} (-1)^{|\lambda|} s_\lambda(Z) \expval{\mathsf{s}_\lambda(X;Y)}_4
    \nonumber \\
    & = \frac{1}{\Delta_M(Z)\mathcal{Z}_{N}^{(4)}} \sum_{\lambda \subseteq (M^{2N})} (-1)^{|\lambda|} \det_{\alpha,\beta \in [M]} z_\alpha^{\lambda_\beta^\vee + M - \beta}
    \Pf_{i,j \in [2N]} 
    \rvev{ x^{\lambda_i + 2N - i} \mid \mathsf{A} \mid y^{\lambda_j + 2N - j} }
    \nonumber \\
    & = \frac{1}{\Delta_M(Z)\mathcal{Z}_{N}^{(4)}} \Pf_{\substack{i,j \in [2N+M] \\ \alpha,\beta \in [M]}}
    \begin{bmatrix}
    \rvev{ x^{2N + M - i} \mid \mathsf{A} \mid y^{2N + M - j} }
    & z_\alpha^{2N+M-i} \\
    - z_\beta^{2N+M-j} & 0
    \end{bmatrix}
    \nonumber \\
    & = \frac{1}{\Delta_M(Z)\mathcal{Z}_{N}^{(4)}} \Pf_{\substack{i,j \in [2N+M] \\ \alpha,\beta \in [M]}}
    \begin{bmatrix}
    \rvev{ p_{2N + M - i} \mid \mathsf{A} \mid p_{2N + M - j} }
    & p_{2N+M-i}(z_\alpha) \\
    - p_{2N+M-j}(z_\beta) & 0
    \end{bmatrix}
    \nonumber \\
    & = \frac{1}{\Delta_M(Z)}\frac{\mathcal{Z}_{N+M/2}^{(4)}}{\mathcal{Z}_{N}^{(4)}} \Pf_{\alpha,\beta \in [M]} \qty[ \sum_{i,j=0}^{2N+M-1} p_{i}(z_\alpha) \widetilde{N}^{(4)}_{i,j} p_j(z_\beta) ]
    = \frac{1}{\Delta_M(Z)} \frac{\mathcal{Z}_{N+M/2}^{(4)}}{\mathcal{Z}_N^{(4)}}
    \Pf_{\alpha,\beta \in[M]} {K}_{N+M/2}^{(4)} (z_\alpha,z_\beta)
    \, .
\end{align}
This completes the proof.
\end{proof}

\begin{corollary}[Characteristic polynomial]
The characteristic polynomial average with respect to the symplectic ensemble is given by a Pfaffian,
\begin{align}
    \expval{\prod_{\alpha \in [M]} \det(z_\alpha - X)^2}_4 = \left. \frac{\mathcal{Z}_{N+M/2}^{(4)}/\mathcal{Z}_N^{(4)}}{\Delta_M(Z)}
    \Pf_{\alpha,\beta \in[M]} {K}_{N+M/2}^{(4)} (z_\alpha,z_\beta)\right|_\eqref{eq:beta4specialization}
\end{align}
\end{corollary}

We then consider the characteristic polynomial inverse for the $\beta = 4$ model.
For this purpose, we introduce the dual CD kernel as follows.
\begin{definition}[Dual Christoffel--Darboux kernel]
Let $N \in \mathbb{N}$.
We define the integral operator as follows,
\begin{align}
    \bar{\mathsf{A}}_N
    = \mathsf{A} \qty( \widetilde{\mathsf{A}}  - \mathsf{K}_{N}^{(4)} ) \mathsf{A}
    \, .
    \label{eq:Abar4}
\end{align}
Recalling the relation~\eqref{eq:K4_infini}, we may write
\begin{align}
    \bar{\mathsf{A}}_N = \sum_{i=N}^\infty \mathsf{A} 
    \begin{bmatrix}
    \ket{\phi_{2i-1}} \bra{\phi_{2i-2}} - \ket{\phi_{2i-2}} \bra{\phi_{2i-1}} &
    \ket{\phi_{2i-1}} \bra{\phi'_{2i-2}} - \ket{\phi_{2i-2}} \bra{\phi_{2i-1}'} \\[.5em]
    \ket{\phi_{2i-1}'} \bra{\phi_{2i-2}} - \ket{\phi_{2i-2}'} \bra{\phi_{2i-1}} &
    \ket{\phi_{2i-1}'} \bra{\phi_{2i-2}'} - \ket{\phi_{2i-2}'} \bra{\phi_{2i-1}'}
    \end{bmatrix}
    \mathsf{A}
    \, ,
\end{align}
where $(\phi_k)_{k \in \mathbb{N}}$ is a set of degree-$k$ skew-orthonormal polynomials, $\rvev{ \phi_{2i-2} \mid \mathsf{A} \mid \phi_{2j-1} } = \delta_{i,j}$, $\rvev{ \phi_{2i} \mid \mathsf{A} \mid \phi_{2j}} = \rvev{\phi_{2i-1} \mid \mathsf{A} \mid \phi_{2j-1}} = 0$.
Then, we define the dual matrix Christoffel--Darboux kernel through the double Hilbert transform,
\begin{align}
    \widetilde{\mathsf{K}}^{(4)}_{N}(z,w) 
    & = 
    \int \dd{\mu(x)} \dd{\mu(\tilde{x})} \frac{ \bar{\mathsf{A}}_N(x,\tilde{x}) }{(z - x)(w - \tilde{x})}
    \, .
\end{align}
We may also define an auxiliary kernel (the $(1,1)$-component of the matrix kernel $\widetilde{\mathsf{K}}^{(4)}_N$),
\begin{align}
    \widetilde{K}^{(4)}_N(z,w) 
    & =
    \int \dd{\mu(x)} \dd{\mu(\tilde{x})} \sum_{i= N}^\infty  \frac{ \bra{x \mid \mathsf{A} } \qty( \ket{\phi_{2i-1}} \bra{\phi_{2i-2}} -  \ket{\phi_{2i-2}} \bra{\phi_{2i-1}} ) \ket{ \mathsf{A} \mid \tilde{x} } }{(z - x)(w - \tilde{x})}
    \, .
\end{align}
\end{definition}

Using this dual CD kernel, we obtain the following Pfaffian formula.
\begin{proposition}
Let $N \in \mathbb{N}$, $M \in 2\mathbb{N}$, and we denote $Z = (z_\alpha)_{\alpha \in [M]}$.
Then, the following Pfaffian formulas hold.
\begin{subequations}
\begin{align}
    & 2N \ge M : 
    \sum_{\ell(\lambda) \le M} \frac{s_\lambda(Z^{-1})}{\det Z^{2N}} \expval{\mathsf{s}_\lambda(X;Y)}_4
    = \frac{1}{\Delta_M(Z)} \frac{\mathcal{Z}_{N-M/2}^{(4)}}{\mathcal{Z}_N^{(4)}}
    \Pf_{\alpha,\beta \in[M]} \widetilde{\mathsf{K}}_{N-M/2}^{(4)} (z_\alpha,z_\beta)
    \, , \\
    & 2N \le M : 
    \sum_{\ell(\lambda) \le 2N} \frac{s_\lambda(Z^{-1})}{\det Z^{2N}} \expval{\mathsf{s}_\lambda(X;Y)}_4 \nonumber \\
    & = \frac{1}{\Delta_M(Z)\mathcal{Z}_N^{(4)}}\Pf_{\alpha,\beta \in [M]} \mathsf{\widetilde{K}}_0^{(4)}(z_\alpha,z_\beta) \Pf_{k,l \in [M-2N]}\qty[\sum_{\alpha,\beta=1}^M \rvev{p_{M-2N-k} \mid z_\alpha}\rvev{z_\alpha \mid \frac{1}{\mathsf{\widetilde{K}}_0^{(4)}} \mid z_\beta} \rvev{z_\beta \mid p_{M-2N-l}}]
    \, .
\end{align}
\end{subequations}
\end{proposition}
\begin{proof}
We first consider the case $2N \ge M$.
The summation over the partition is given as follows,
\begin{align}
    & \sum_{\ell(\lambda) \le M} \frac{s_\lambda(Z^{-1})}{\det Z^{2N}} \expval{\mathsf{s}_\lambda(X;Y)}_4
    \nonumber \\
    & = \frac{1}{\Delta_{M}(Z) \mathcal{Z}_N^{(4)}} \sum_{\ell(\lambda) \le M} \det_{\alpha,\beta \in [M]} z_\alpha^{-\lambda_\beta - 2N + \beta - 1}
    \nonumber \\
    & \hspace{5em} \times
    \Pf_{\substack{\alpha,\beta \in [M] \\ k,l \in [2N-M]}}
    \begin{bmatrix}
    \rvev{ x^{\lambda_\alpha + 2N - \alpha} \mid \mathsf{A} \mid y^{\lambda_\beta + 2N - \beta} } &
    \rvev{ x^{\lambda_\alpha + 2N - \alpha} \mid \mathsf{A} \mid y^{2N - M - l} } \\
    \rvev{ x^{2N-M-k} \mid \mathsf{A} \mid y^{\lambda_\beta + 2N - \beta} } &
    \rvev{ x^{2N-M-k} \mid \mathsf{A} \mid y^{2N-M-l} }
    \end{bmatrix}
    \nonumber \\
    & =
    \frac{1}{\Delta_{M}(Z)} \frac{\mathcal{Z}_{N-M/2}^{(4)}}{\mathcal{Z}_N^{(4)}}
    \frac{1}{M!} \sum_{ \substack{ r_\alpha = 0 ,\ldots, \infty \\ r_\alpha \neq r_\beta } } \det_{\alpha,\beta \in [M]} z_\alpha^{-r_\beta + M - 2N - 1}
    \Pf_{\alpha,\beta \in [M]} \rvev{ x^{r_\alpha - M + 2N} \mid \bar{\mathsf{A}}_{N-M/2} \mid y^{r_\beta - M + 2N} }
    \nonumber \\
    & = 
    \frac{1}{\Delta_{M}(Z)} \frac{\mathcal{Z}_{N-M/2}^{(4)}}{\mathcal{Z}_N^{(4)}}
    \Pf_{\alpha,\beta \in [M]}
    \qty[ \sum_{r,s=0}^\infty z_\alpha^{-r+M-2N-1} z_\beta^{-s+M-2N-1} \rvev{ x^{r - M + 2N} \mid \bar{\mathsf{A}}_{N-M/2} \mid y^{s - M + 2N} } ]
    \nonumber \\
    & = 
    \frac{1}{\Delta_{M}(Z)} \frac{\mathcal{Z}_{N-M/2}^{(4)}}{\mathcal{Z}_N^{(4)}}
    \Pf_{\alpha,\beta \in [M]} \widetilde{\mathsf{K}}_{N-M/2}^{(4)} (z_\alpha,z_\beta)
    \, .
\end{align}

In order to obtain the last expression, we may apply the same argument to the $\beta = 1$ case~\eqref{eq:sum_to_Hilbert}. The case $2N \le M$ can be obtained in a similar way. Indeed,

\begin{align}
    &\sum_{\ell(\lambda) \le 2N} \frac{s_\lambda(Z^{-1})}{\det Z^{2N}} \expval{\mathsf{s}_\lambda(X;Y)}_4 \nonumber \\
    &= \frac{1}{\mathcal{Z}_N^{(4)} \Delta_M(Z)} \sum_{\ell(\lambda) \le 2N} \det_{\substack{\alpha \in [M], \beta \in [2N] \\ k \in [M-2N]}} 
    \begin{bmatrix}
    z_\alpha^{- \lambda_\beta - 2N + \beta - 1} \\ p_{M - 2N - k}(z_\alpha)
    \end{bmatrix}
    \Pf_{\alpha,\beta \in [2N]}
    \rvev{ x^{\lambda_\alpha + 2N - \alpha} \mid \mathsf{A} \mid y^{\lambda_\beta + 2N - \beta} } \nonumber \\
    &= \frac{1}{\mathcal{Z}_N^{(4)} \Delta_M(Z)} \Pf_{\substack{\alpha,\beta \in [M]\\k,l \in [M-2N]}} 
    \begin{bmatrix}
    \mathsf{\widetilde{K}}_0^{(4)}(z_\alpha,z_\beta) & p_{M-2N-l}(z_\alpha) \\
    -p_{M-2N-k}(z_\beta) & 0
    \end{bmatrix} \nonumber \\
    &= \frac{1}{\Delta_M(Z)\mathcal{Z}_N^{(4)}}\Pf_{\alpha,\beta \in [M]} \mathsf{\widetilde{K}}_0^{(4)}(z_\alpha,z_\beta) \Pf_{k,l \in [M-2N]}\qty[\sum_{\alpha,\beta=1}^M \rvev{p_{M-2N-k} \mid z_\alpha}\rvev{z_\alpha \mid \frac{1}{\mathsf{\widetilde{K}}_0^{(4)}} \mid z_\beta} \rvev{z_\beta \mid p_{M-2N-l}}]
\end{align}
where first the Cauchy-Binet-type formula \eqref{eq:CB0} and then the Pfaffian of a block matrix  \eqref{eq:Pf_block} have been used. This completes the proof. 
\end{proof}
\begin{corollary}[Characteristic polynomial inverse] 
The characteristic polynomial inverse with respect to the symplectic ensemble is therefore given by the Pfaffians 
\begin{subequations}
\begin{align}
    & 2N \ge M : 
    \expval{ \prod_{\alpha \in [M]} \det(z_\alpha - X)^{-2} }_4
    = \left. \frac{1}{\Delta_M(Z)} \frac{\mathcal{Z}_{N-M/2}^{(4)}}{\mathcal{Z}_N^{(4)}}
    \Pf_{\alpha,\beta \in[M]} \widetilde{\mathsf{K}}_{N-M/2}^{(4)} (z_\alpha,z_\beta) \right|_\eqref{eq:beta4specialization}
    \, , \\
    & 2N \le M : 
    \,  \expval{ \prod_{\alpha \in [M]} \det(z_\alpha - X)^{-2} }_4 \nonumber \\
    & = \frac{1}{\Delta_M(Z)\mathcal{Z}_N^{(4)}}\Pf_{\alpha,\beta \in [M]} \mathsf{\widetilde{K}}_0^{(4)}(z_\alpha,z_\beta) \Pf_{k,l \in [M-2N]}\qty[\sum_{\alpha,\beta=1}^M \rvev{p_{M-2N-k} \mid z_\alpha}\rvev{z_\alpha \mid \frac{1}{\mathsf{\widetilde{K}}_0^{(4)}} \mid z_\beta} \rvev{z_\beta \mid p_{M-2N-l}}].
\end{align}
\end{subequations}
\end{corollary}

\section{$BCD$-quiver matrix models}\label{sec:chain}

We consider more generalized situations, that we call quiver matrix models, involving several sets of formal eigenvalues with interactions.
We in particular explore the $BCD$-type quivers.

\subsection{$D$-type quiver models}\label{sec:D-model}

We first consider the matrix model that we call the $D$-type quiver model as follows.
\begin{definition}
Let $L, N \in \mathbb{N}$ and $(X_k,\widetilde{X}_k)_{k \in [L]} = (x_{k,i},\tilde{x}_{k,i})_{i \in [N]}^{k \in [L]}$ a set of formal eigenvalues.
Denoting $Z_k = (z_{k,i},z_{k,N+i})_{i \in [N]} = (x_{k,i},\tilde{x}_{k,i})_{i \in [N]}$ and $\underline{X} = (X_k, \widetilde{X}_k)_{k \in [L]}$, we define the partition function of the $D_{L+1}$-type quiver matrix model,
\begin{align}
    &
    \mathcal{Z}[D_{L+1}] 
    \nonumber \\
    & = \int \frac{\dd{\mu(\underline{X}})}{(2N)!^{L-1} N!^2}
    \frac{\Delta_{2N}(Z_1)^2 \cdots \Delta_{2N}(Z_{L-1})^2 \Delta_N(X_L)^2 \Delta_{N}(\widetilde{X}_L)^2}{\displaystyle  \prod_{i,j \in [2N]} (z_{1,i} - z_{2,j}) \cdots (z_{L-2,i} - z_{L-1,j}) \prod_{i\in[2N],j\in[N]} (z_{L-1,i} - x_{L,j}) (z_{L-1,i} - \tilde{x}_{L,j}) }
    \, .
    \label{eq:D-model}
\end{align}
\end{definition}
This model is graphically described using the quiver diagram of type $D_{L+1}$ as shown in \eqref{eq:Dynkin_BCD}.
For this $D$-type model, the following Pfaffian formula holds.
\begin{proposition}
We denote $\displaystyle \omega(x,y) = \frac{1}{x-y}$.
Then, the following Pfaffian formula holds for the partition function of the $D$-type matrix model,
\begin{align}
    \mathcal{Z}[D_{L+1}]
    & =  \Pf_{i, j \in [2N]}
    \Bigg[ \int \prod_{k=1}^{L} \dd{\mu(x_k;\tilde{x}_k)}
    \Bigg( \frac{p_{2N-i}(x_1)}{x_1-x_2} \cdots \frac{1}{x_{L-1} - x_L} \frac{1}{x_L - \tilde{x}_L} \frac{1}{\tilde{x}_L - \tilde{x}_{L-1}} \cdots \frac{p_{2N-j}(\tilde{x}_1)}{\tilde{x}_2 - \tilde{x}_1}  
    \nonumber \\ & \hspace{8em} 
    - \frac{p_{2N-j}(\tilde{x}_1)}{\tilde{x}_1-\tilde{x}_2} \cdots \frac{1}{\tilde{x}_{L-1} - \tilde{x}_L} \frac{1}{\tilde{x}_L - {x}_L} \frac{1}{{x}_L - {x}_{L-1}} \cdots \frac{p_{2N-i}({x}_1)}{{x}_2 - {x}_1} 
    \Bigg)\Bigg]
    \nonumber \\
    & =: \Pf_{i, j \in [2N]}
    \qty[ \rvev{p_{2N-i} \mid \omega^{2L-1} \mid p_{2N-j}} - \rvev{p_{2N-j} \mid \omega^{2L-1} \mid p_{2N-i}} ]  
    \, .
\end{align}
\end{proposition}
\begin{proof}
Recalling the relation
\begin{align}
    \Delta_N(X_L) \Delta_N(\widetilde{X}_L) 
    = \frac{\Delta_{2N}(Z_L)}{\prod_{i,j \in [N]} (x_{L,i} - \tilde{x}_{L,j} ) }
    \, ,
\end{align}
we may rewrite the partition function as follows,
\begin{align}
    \mathcal{Z}[D_{L+1}]
    & = \int \frac{\dd{\mu(\underline{X}})}{(2N)!^{L-1} N!^2}
    \frac{\Delta_{2N}(Z_1)^2 \cdots \Delta_{2N}(Z_{L-1})^2  \Delta_{2N}(Z_L)^2}{\displaystyle  \prod_{i,j \in [2N]} (z_{1,i} - z_{2,j}) \cdots (z_{L-1,i} - z_{L,j}) \prod_{i,j\in[N]} (x_{L,i} - \tilde{x}_{L,j}) }
    \nonumber \\
    & = 
    \int \frac{\dd{\mu(\underline{X}})}{(2N)!^{L-1} N!^2}
    \det_{i,j \in [2N]} p_{2N-i}(z_{1,j}) 
    \prod_{k=1}^{L-1} \qty[ \det_{i,j \in [2N]} \qty(\frac{1}{z_{k,i} - z_{k+1,j}}) ] 
    \det_{i,j \in [N]} \qty( \frac{1}{x_{L,i} - \tilde{x}_{L,j}} )
    \, .
\end{align}
Applying Andréief formula~\eqref{eq:A_formula} and de Bruijn's formula~\eqref{eq:dB_formula} recursively, we obtain 
\begin{align}
    \mathcal{Z}[D_{L+1}]
    & = \int \frac{\dd{\mu(Z_L)}}{N!^2}
    \det_{ \substack{i\in [2N] \\ j \in [N]} } 
    \begin{bmatrix}
    \rvev{p_{2N-i} \mid \omega^{L-1} \mid x_{L,j}} \\
    \rvev{p_{2N-i} \mid \omega^{L-1} \mid \tilde{x}_{L,j}}
    \end{bmatrix}
    \det_{i,j \in [N]} \rvev{x_{L,i} \mid \omega \mid \tilde{x}_{L,j}}
    \nonumber \\
    & = \int \frac{\dd{\mu(Z_L)}}{N!^2}
    \det_{ \substack{i\in [2N] \\ j \in [N]} } 
    \begin{bmatrix}
    \rvev{p_{2N-i} \mid \omega^{L-1} \mid x_{L,j}} \\
    \rvev{p_{2N-i} \mid \omega^{L-1} \mid \tilde{x}_{L,j}}
    \end{bmatrix}
    \Pf_{i,j \in [N]}
    \begin{bmatrix}
    0 & \rvev{x_{L,i} \mid \omega \mid \tilde{x}_{L,j}} \\
    - \rvev{\tilde{x}_{L,i} \mid \omega \mid {x}_{L,j}} & 0
    \end{bmatrix}
    \nonumber \\
    & = \Pf_{i, j \in [2N]}
    \qty[ \rvev{p_{2N-i} \mid \omega^{2L-1} \mid p_{2N-j}} - \rvev{p_{2N-j} \mid \omega^{2L-1} \mid p_{2N-i}} ]
    \, .
\end{align}
This completes the proof.
\end{proof}

\begin{remark}
We may apply the same approach to consider the average of the characteristic polynomial coupled with the left-most matrix $X_1$, and obtain the Pfaffian formula in terms of the corresponding CD kernel.
See also \cite{Babinet:2022yij} for the $A$-type quiver model calculation.
\end{remark}

\subsection{Matrix chain with the Pfaffian interaction}

We then consider the matrix chain generalization of the $\beta = 1$ and $\beta = 4$ models.

\begin{definition}
Let $N \in 2\mathbb{N} \ (\beta = 1), \mathbb{N} \ (\beta = 4)$, $L \in \mathbb{N}$ and $\underline{X} = (X_k)_{k \in [L]} = (x_{k,i})_{i \in [N]}^{k \in [L]}$ a set of formal eigenvalues.
Let $A(x,y)$, $B(x,y)$ and $S(x,y)$ be skew-symmetric and generic integrable functions as before.
We denote by $(\omega_{k}(x,y))_{k=1,\ldots,L-1}$ generic integrable two-variable functions, and by $(f_{i}(x),g_{i}(x))_{i \in \mathbb{Z}_{\ge 0}}$ a set of integrable functions.
Then, we define the partition function of the $\beta = 1$ and $\beta = 4$ matrix chain models as follows,
\begin{subequations}
\begin{align}
    \mathcal{Z}_{N^L}^{(1)} & = \frac{1}{N!^L} \int \dd{\mu(\underline{X})} \det_{i,j \in [N]} f_{j-1}(x_{1,i}) 
    \prod_{k=1}^{L-1} \qty[ \det_{i,j \in [N]} \omega_{k}(x_{k,i},x_{k+1,j}) ]
    \Pf_{i,j \in [N]} A(x_{L,i},x_{L,j})
    \, , \\
    \mathcal{Z}_{N^L}^{(4)} & = \frac{1}{N!^{2L}} \int \dd{\mu(\underline{X})} \det_{i \in [N], j \in [2N]} \qty[ f_{j-1}(x_{1,i}) \ g_{j-1}(\tilde{x}_{1,i}) ]
    \nonumber \\
    & \hspace{5em} \times
    \prod_{k=1}^{L-1} \qty[ \det_{i,j \in [N]} {\Omega}_{k}(x_{k,i},\tilde{x}_{k,i};x_{k+1,j},\tilde{x}_{k+1,j}) ]
    \Pf_{i,j \in [N]} \mathsf{A}(x_{L,i},\tilde{x}_{L,i};x_{L,j},\tilde{x}_{L,j})
\end{align}
where we write
\begin{align}
    {\Omega}_{k}(x,\tilde{x};y,\tilde{y}) = 
    \begin{bmatrix}
    \omega_{k}^{11}(x,y) & \omega_{k}^{12}(x,\tilde{y}) \\
    \omega_{k}^{21}(\tilde{x},y) & \omega_{k}^{22}(\tilde{x},\tilde{y})
    \end{bmatrix}
    \, , \qquad
    \mathsf{A}(x,\tilde{x};y,\tilde{y}) = 
    \begin{bmatrix}
    A(x,y) & S(x,\tilde{y}) \\
    -S^\text{T}(\tilde{x},y) & B(\tilde{x},\tilde{y})
    \end{bmatrix}
    \, .
\end{align}
\end{subequations}
\end{definition}
For these partition functions, we obtain the Pfaffian formulas as follows.
\begin{proposition}\label{prop:Z_chain}
We define the dressed integral operators,
\begin{subequations}
\begin{align}
    A_{L-1} & = \omega_{1} \cdots \omega_{L-1} \cdot A \cdot \omega_{L-1}^\text{T} \cdots \omega_1^\text{T}
    \, , \\
    \mathsf{A}_{L-1} & = \Omega_1 \cdots \Omega_{L-1} \cdot \mathsf{A} \cdot \Omega_{L-1}^\text{T} \cdots \Omega_1^\text{T}
    \, .
 \end{align}
 \end{subequations}
Then, the partition functions of the coupled $\beta = 1$ and $\beta = 4$ models are given by Pfaffians,
\begin{subequations}
\begin{align}
    \mathcal{Z}_{N^L}^{(1)} & = \Pf_{i,j \in [N]} \rvev{ f_{i-1} \mid A_{L-1} \mid f_{j-1} }
    \, , \\
    \mathcal{Z}_{N^L}^{(4)} & = \Pf_{i,j \in [N]} \rvev{ f_{i-1}, g_{i-1} \mid \mathsf{A}_{L-1} \mid f_{j-1}, g_{j-1} }
    \, .
\end{align}
\end{subequations}
\end{proposition}
\begin{proof}
We first consider the case $\beta = 1$.
Applying Andréief's formula~\eqref{eq:A_formula}, we integrate the eigenvalues $(x_{k,i})_{i \in [N]}^{k=1,\ldots,L-1}$ as follows,
\begin{align}
    &
    \frac{1}{N!^L} \int \dd{\mu(\underline{X})} \det_{i,j \in [N]} f_{j-1}(x_{1,i}) 
    \prod_{k=1}^{L-1} \qty[ \det_{i,j \in [N]} \omega_{k}(x_{k,i},x_{k+1,j}) ]
    \nonumber \\
    & = \frac{1}{N!} \int \dd{\mu(X_L)} \det_{i,j \in [N]} \rvev{f_{j-1} \mid \omega_1 \cdots \omega_{L-1} \mid x_{L,i}}
    \, .
\end{align}
We have a similar relation for the case $\beta = 4$,
\begin{align}
    & \frac{1}{N!^{2L}} \int \dd{\mu(\underline{X})} \det_{i \in [N], j \in [2N]} \qty[ f_{j-1}(x_{1,i}) \ g_{j-1}(\tilde{x}_{1,i}) ]
    \prod_{k=1}^{L-1} \qty[ \det_{i,j \in [N]} {\Omega}_{k}(x_{k,i},\tilde{x}_{k,i};x_{k+1,j},\tilde{x}_{k+1,j}) ]
    \nonumber \\
    & = \frac{1}{N!^{2}} \int \dd{\mu(X_L)} \det_{i \in [N], j \in [2N]} \qty[
    \begin{bmatrix}
    \bra{f_{j-1}} \ \bra{g_{j-1}}
    \end{bmatrix}
    \Omega_1 \cdots \Omega_{L-1}
    \begin{bmatrix}
    \ket{x_{L,i}} \\ \ket{\tilde{x}_{L,i}}
    \end{bmatrix}
    ]
    \, .
\end{align}
Then, we apply de Bruijn's formula~\eqref{eq:dB_formula} to integrate the remaining variables $(x_{L,i})_{i \in [N]}$, and obtain the result.
\end{proof}

These generic models yield the following special cases that would have quiver interpretations.

\subsubsection{$B$-type quiver models}\label{sec:B-quiver}

We have the following examples for the matrix chain models.
For the $\beta = 4$ model, we consider
\begin{gather}
    f_{i-1}(x) = g_{i-1}(x) = p_{2N-i}(x) \, , \quad
    A(x,y) = B(x,y) = 0 \, , \quad
    S(x,y) = \delta(x-y) \, , \nonumber \\
    \omega_k^{l,m}(x,y) = 
    \frac{1}{x-y}  \quad (k=1,\ldots,L-2)
    \, . \qquad
    \Omega_{L-1}(x,\tilde{x};y,\tilde{y}) = 
    \begin{bmatrix}
    \frac{1}{x-y} & \frac{1}{(x-\tilde{y})^2} \\
    \frac{1}{\tilde{x} - y} & \frac{1}{(\tilde{x} - \tilde{y})^2}
    \end{bmatrix}
    \, ,
\end{gather}
and denote $Z_k = (z_{k,i},z_{k,N+i})_{i \in [N]} = (x_{k,i},\tilde{x}_{k,i})_{i \in [N]}$ and $\underline{X} = (X_k)_{k \in [L]}$.
Then, we obtain
\begin{align}
    \mathcal{Z}_{N^L}^{(4)} = 
    \int \frac{\dd{\mu(\underline{X})}}{(2N)!^{L-1} N!} 
    \frac{\Delta_{2N}(Z_1)^2 \cdots \Delta_{2N}(Z_{L-1})^2 \Delta_N(X_L)^4}{\displaystyle  \prod_{i,j \in [2N]} (z_{1,i} - z_{2,j}) \cdots (z_{L-2,i} - z_{L-1,j}) \prod_{i\in[2N],j\in[N]} (z_{L-1,i} - x_{L,j})^2 }
    \, ,
    \label{eq:B-model}
\end{align}
where we change the normalization of the partition function.

In fact, this model may have an interpretation as the $B$-type quiver matrix model, which is constructed from the $D$-type model~\eqref{eq:D-model}: 
Under the process called the folding $X_L = \widetilde{X}_L$, we obtain the partition function~\eqref{eq:B-model}.
This situation is graphically explained using the quiver diagrams:
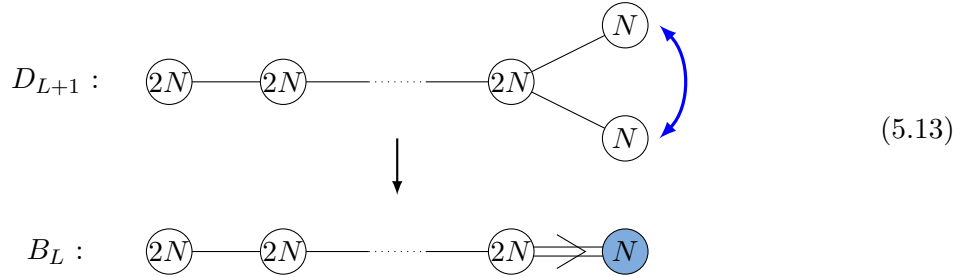
\begin{align}
    \begin{tikzpicture}[baseline=(current bounding box.center),scale=1.5]
    \node at (-1,0) {$D_{L+1}$ :};
    \draw (0,0) -- (1.75,0);
    \draw[dotted] (1.75,0) -- (2.25,0);
    \draw (2.25,0) -- (3,0);
    \draw (4,.5) -- (3,0) -- (4,-.5);
    \filldraw[fill = white, draw = black] (0,0) circle (.2) node {$2N$};
    \filldraw[fill = white, draw = black] (1,0) circle (.2) node {$2N$};
    \filldraw[fill = white, draw = black] (3,0) circle (.2) node {$2N$};
    \filldraw[fill = white, draw = black] (4,.5) circle (.2) node {$N$};
    \filldraw[fill = white, draw = black] (4,-.5) circle (.2) node {$N$};    
    \draw[latex-latex,blue,very thick] (4.3,.5) to [out=-45,in=45] ++(0,-1); 
    \draw[-latex,thick] (2,-.5) -- ++(0,-.5);
    \begin{scope}[shift={(0,-1.5)}]
    \node at (-1,0) {$B_{L}$ :};
    \draw (0,0) -- (1.75,0);
    \draw[dotted] (1.75,0) -- (2.25,0);
    \draw (2.25,0) -- (3,0);
    \draw[double,double distance = 3pt] (3,0) -- (4,0);
    \draw (3.4,.15) -- ++(.25,-.15) -- ++(-.25,-.15);    
    \filldraw[fill = white, draw = black] (0,0) circle (.2) node {$2N$};
    \filldraw[fill = white, draw = black] (1,0) circle (.2) node {$2N$};
    \filldraw[fill = white, draw = black] (3,0) circle (.2) node {$2N$};
    \filldraw[fill = c1!50, draw = black] (4,0) circle (.2) node {$N$};
    \end{scope}
    \end{tikzpicture}
    \label{eq:DtoB}
\end{align}
In these diagrams, we denote the $\beta$-node $\Delta_N(X)^\beta$ by 
\tikz[baseline=-4pt,scale=1.2] \filldraw[fill = c2!50,draw = black] (0,0) circle (.2) node {$N$}; ($\beta = 1$),%
\footnote{%
We will consider the $\beta = 1$ node in the next part.
}
\tikz[baseline=-4pt,scale=1.2] \draw (0,0) circle (.2) node {$N$}; ($\beta = 2$),
\tikz[baseline=-4pt,scale=1.2] \filldraw[fill = c1!50,draw = black] (0,0) circle (.2) node {$N$}; ($\beta = 4$), and the line connecting the nodes indicates the Cauchy-type interaction $\prod_{i,j \in [N]} (x_{k,i} - x_{k+1,j})^{-1}$~\cite{Bertola:2009CMP}.

\subsubsection{$C$-type quiver models}\label{sec:C-quiver}

For the $\beta = 1$ model, putting
\begin{align}
    f_{i-1}(x) = p_{N-i}(x) \, , \qquad
    \omega_k(x,y) = \frac{1}{x-y} \, , \qquad
    A(x,y) = \operatorname{sgn}(x-y) \, ,
\end{align}
we obtain
\begin{align}
    \mathcal{Z}_{N^L}^{(1)} = 
    \int \frac{\dd{\mu(X^L)}}{N!^L} 
    \frac{\Delta_N(X_1)^2 \cdots \Delta_N(X_{L-1})^2 |\Delta_N(X_L)|}{\displaystyle \prod_{i,j \in [N]} (x_{1,i} - x_{2,j}) \cdots (x_{L-1,i} - x_{L,j})} 
    \, .
    \label{eq:C-model}
\end{align}
This model would have an interpretation as the $C$-type model obtained from the $A$-type model through the folding.
We consider the quiver matrix model of the type $A_{2L-1}$,
\begin{align}
    \mathcal{Z}[A_{2L-1}] = \int \dd{\mu(\underline{X})} 
    \frac{\Delta_N(X_1)^2 \cdots \Delta_N(X_{2L-1})^2}{\displaystyle \prod_{i,j \in [N]} (x_{1,i} - x_{2,j}) \cdots (x_{2L-2,i} - x_{2L-1,j})} 
    \, .
\end{align}
Identifying $X_{k} = X_{2L-k}$ for $k = 1,\ldots,L-1$, we obtain
\begin{align}
    \int \dd{\mu(\underline{X})} 
    \frac{\Delta_N(X_1)^4 \cdots \Delta_N(X_{L-1})^4 \Delta_N(X_L)^2}{\displaystyle \prod_{i,j \in [N]} (x_{1,i} - x_{2,j})^2 \cdots (x_{L-1,i} - x_{L,j})^2} 
    \, .
\end{align}
Then, formally taking a square root of the integrand, we obtain the partition function shown in \eqref{eq:C-model}.
This situation is depicted as follows:
\begin{align}
    \begin{tikzpicture}[baseline=(current bounding box.center),scale=1.5]
    \node at (-1,0) {$A_{2L-1}$ :};
    \draw (0,.5) -- (1.75,.5);
    \draw[dotted] (1.75,.5) -- (2.25,.5);
    \draw (2.25,.5) -- (3,.5) -- (4,0) -- (3,-.5) -- (2.25,-.5);
    \draw (0,-.5) -- (1.75,-.5);
    \draw[dotted] (1.75,-.5) -- (2.25,-.5);
    \filldraw[fill = white, draw = black] (0,.5) circle (.2) node {$N$};
    \filldraw[fill = white, draw = black] (1,.5) circle (.2) node {$N$};
    \filldraw[fill = white, draw = black] (3,.5) circle (.2) node {$N$};
    \filldraw[fill = white, draw = black] (4,0) circle (.2) node {$N$};
    \filldraw[fill = white, draw = black] (0,-.5) circle (.2) node {$N$};
    \filldraw[fill = white, draw = black] (1,-.5) circle (.2) node {$N$};
    \filldraw[fill = white, draw = black] (3,-.5) circle (.2) node {$N$};    
    \draw[latex-latex,blue,very thick] (0,.25) -- ++(0,-.5); 
    \draw[latex-latex,blue,very thick] (1,.25) -- ++(0,-.5); 
    \draw[latex-latex,blue,very thick] (3,.25) -- ++(0,-.5); 
    \draw[-latex,thick] (2,-1) -- ++(0,-.5);
    \begin{scope}[shift={(0,-2)}]
    \node at (-1,0) {$C_{L}$ :};
    \draw (0,0) -- (1.75,0);
    \draw[dotted] (1.75,0) -- (2.25,0);
    \draw (2.25,0) -- (3,0);
    \draw[double,double distance = 3pt] (3,0) -- (4,0);
    \draw (3.6,.15) -- ++(-.25,-.15) -- ++(.25,-.15);    
    \filldraw[fill = white, draw = black] (0,0) circle (.2) node {$N$};
    \filldraw[fill = white, draw = black] (1,0) circle (.2) node {$N$};
    \filldraw[fill = white, draw = black] (3,0) circle (.2) node {$N$};
    \filldraw[fill = c2!50, draw = black] (4,0) circle (.2) node {$N$};
    \end{scope}
    \end{tikzpicture}
\end{align}
We may apply the Pfaffian formula (Proposition~\ref{prop:Z_chain}) to obtain the partition function of this quiver matrix model, and also compute the characteristic polynomial average coupled to the left-most matrix $X_1$ as discussed in Section~\ref{sec:ch_poly}.

\subsection{Generic quiver matrix model}

Summarizing the results above, we establish the rule to construct generic quiver matrix models.
For the oriented arrows and the trivalent vertex, we assign the following interaction terms,
\begin{subequations}
\begin{align}
    \begin{tikzpicture}[baseline=(current bounding box.center)]
    \draw (-.25,.2) -- ++(-.5,0);
    \draw [dotted] (-.75,.2) -- ++(-.5,0);
    \node at (0,0) {\dynkin[edge length=.8cm,labels={a,b}]{B}{.oo}};
    \end{tikzpicture}
    \ & : \quad
    \frac{\Delta_{N_{a}}(X_{a})^2 \Delta_{N_b}(X_{b})^4}{\prod_{i,j} (x_{a,i} - x_{b,j})^2}
    \\[1em]
    \begin{tikzpicture}[baseline=(current bounding box.center)]
    \draw (-.25,.2) -- ++(-.5,0);
    \draw [dotted] (-.75,.2) -- ++(-.5,0);
    \node at (0,0) {\dynkin[edge length=.8cm,labels={a,b}]{C}{.oo}};
    \end{tikzpicture}
    \ & : \quad
    \frac{\Delta_{N_{a}}(X_{a})^2 |\Delta_{N_b}(X_{b})|}{\prod_{i,j} (x_{a,i} - x_{b,j})}
    \\[1em]
    \begin{tikzpicture}[baseline=(current bounding box.center)]
    \draw (0,0) -- ++(-.5,0);
    \draw [dotted] (-.5,0) -- ++(-.5,0);
    \draw (60:.6) node [right=.2cm] {{\scriptsize $b$}} -- (0,0) node [right=.2cm] {{\scriptsize $a$}} -- (-60:.6) node [right=.2cm] {{\scriptsize $c$}};
    \filldraw [fill = white, draw = black] (0,0) circle (.08);
    \filldraw [fill = white, draw = black] (60:.6) circle (.08);
    \filldraw [fill = white, draw = black] (-60:.6) circle (.08);
    \end{tikzpicture}
    \ & : \quad
    \frac{\Delta_{N_a}(X_{a})^2 \Delta_{N_b}(X_{b})^2 \Delta_{N_c}(X_{c})^2 }{\prod_{i,j} (x_{a,i} - x_{b,j}) (x_{a,i} - x_{c,j})}
\end{align}
\end{subequations}
In order to write down the quiver matrix model partition function associated with a Dynkin-quiver diagram $\Gamma$, we prepare the notations.
Let $(\alpha_a)_{a \in \gamma}$ be a set of the simple roots associated with the Dynkin-quiver $\Gamma$.
We define the Cartan matrix $(c_{ab})_{a,b \in \Gamma}$ and its symmetrization $(b_{ab})_{a,b \in \Gamma}$ as follows,
\begin{subequations}
\begin{align}
    c_{ab} & = (\alpha_a^\vee,\alpha_b) = \frac{(\alpha_a,\alpha_b)}{(\alpha_a,\alpha_b)} \, , \\
    b_{ab} & = (\alpha_a,\alpha_b) = d_a c_{ab} \, ,\\
    d_{a} & = (\alpha_a,\alpha_a) \, .
\end{align}
\end{subequations}
We also define the factored version of the symmetrized Cartan matrix,
\begin{align}
    \bar{b}_{ab} = \frac{1}{\operatorname{gcd}(b_{ab})} b_{ab} 
    \, .
\end{align}
Then, for a Dynkin-quiver diagram $\Gamma$, we may write the partition function of the quiver matrix model as follows,
\begin{align}
    Z[\Gamma] = \int \dd{\mu(\underline{X})} \prod_{a,b \in \Gamma} \prod_{\substack{ i \in [N_a], j \in [N_b] }} (x_{i,a} - x_{j,b})^{\bar{b}_{ab}(\Gamma^\vee)}
    \, ,
    \label{eq:quiver_model}
\end{align}
where $\bar{b}_{ab}(\Gamma^\vee)$ is the factored symmetrized Cartan matrix for the Langlands dual quiver $\Gamma^\vee$.
For the simply-laced quivers $\Gamma^\vee = \Gamma$, this reproduces the well-known construction~\cite{Marshakov:1991gc,Kharchev:1992iv} as $\bar{b}_{ab} = c_{ab}$.
From this point of view, the duality between $\beta = 1$ and $\beta = 4$ could be interpreted as the Feigin--Frenkel duality~\cite{Feigin:1991wy}.
See \cite{Frenkel:2005pa} for more details along this direction.

For example, for the $B$ and $C$ types, the symmetrization of the Cartan matrix is given as follows,
\begin{subequations}
\begin{align}
    B_r \ : \
    b = 
    \begin{bmatrix}
    2 &&&&\\
    & 2 &&& \\
    && \ddots && \\
    &&& 2 & \\
    &&&& 1
    \end{bmatrix}
    \begin{bmatrix}
    2 & -1 & & & \\
    -1 & 2 & -1 & & \\
    & \ddots & \ddots & \ddots & \\
    && -1 & 2 & -1 \\
    &&& - 2 & 2
    \end{bmatrix}
    = 
    \begin{bmatrix}
    4 & -2 & & & \\
    -2 & 4 & -2 & & \\
    & \ddots & \ddots & \ddots & \\
    && -2 & 4 & -2 \\
    &&& - 2 & 2
    \end{bmatrix}   
    \\
    C_r \ : \
    b = 
\begin{bmatrix}
    1 &&&&\\
    & 1 &&& \\
    && \ddots && \\
    &&& 1 & \\
    &&&& 2
    \end{bmatrix}
    \begin{bmatrix}
    2 & -1 & & & \\
    -1 & 2 & -1 & & \\
    & \ddots & \ddots & \ddots & \\
    && -1 & 2 & -2 \\
    &&& - 1 & 2
    \end{bmatrix}
    = 
    \begin{bmatrix}
    2 & -1 & & & \\
    -1 & 2 & -1 & & \\
    & \ddots & \ddots & \ddots & \\
    && -1 & 2 & -2 \\
    &&& - 2 & 4
    \end{bmatrix}       
\end{align}
\end{subequations}
Hence, the factored matrix is given by
\begin{align}
    \bar{b} = 
    \begin{cases}
    b/2 & (\Gamma = B_r) \\
    b & (\Gamma = C_r)
    \end{cases}
    \, .
\end{align}
Recalling $B_r^\vee = C_r$ and vice versa, the partition function defined in~\eqref{eq:quiver_model} reproduces the expressions \eqref{eq:B-model} and \eqref{eq:C-model}.

\bibliographystyle{amsalpha_mod}
\bibliography{ref}

\end{document}